\documentclass[pdflatex,sn-mathphys-num]{sn-jnl}% Math and Physical Sciences Numbered Reference Style
\usepackage{graphicx}%
\usepackage{multirow}%
\usepackage{amsmath,amssymb,amsfonts}%
\usepackage{amsthm}%
\usepackage{mathrsfs}%
\usepackage[title]{appendix}%
\usepackage{xcolor}%
\usepackage{textcomp}%
\usepackage{manyfoot}%
\usepackage{booktabs}%
\usepackage{algorithm}%
\usepackage{algorithmicx}%
\usepackage{algpseudocode}%
\usepackage{listings}%
\usepackage{graphicx}
\usepackage{subcaption}
\usepackage{placeins}
\theoremstyle{thmstyleone}%
\newtheorem{theorem}{Theorem}%  meant for continuous numbers
\theoremstyle{thmstyletwo}%

\theoremstyle{thmstylethree}%

\begin{document}

\title[Article Title]{Adaptive information-maximization encoding for ghost imaging--A general Bayesian framework under experimental physical constraints}


\author[1,2,3]{\fnm{Jianshuo} \sur{Sun}}

\author*[1]{\fnm{Chenyu} \sur{Hu}}\email{huchenyu@ucas.ac.cn}

\author[2,3]{\fnm{Zunwang} \sur{Bo}}

\author[2,3]{\fnm{Zhentao} \sur{Liu}}

\author[2,3]{\fnm{Mengyu} \sur{Chen}}

\author[4]{\fnm{Longkun	} \sur{Du}}

\author[5]{\fnm{Weitao} \sur{Liu}}

\author*[1,2,3]{\fnm{Shensheng} \sur{Han}}\email{sshan@mail.shcnc.ac.cn}

\affil*[1]{\orgdiv{School of Physics and Optoelectronic Engineering}, \orgname{ Hangzhou Institute for Advanced Study, University of Chinese Academy of Sciences}, \orgaddress{\city{Hangzhou}, \postcode{310024}, \state{Zhejiang}, \country{China}}}

\affil[2]{\orgdiv{Aerospace Laser Technology and System Department}, \orgname{Shanghai Institute of Optics and Fine Mechanics, Chinese Academy of Sciences}, \orgaddress{\city{Shanghai}, \postcode{201800}, \country{China}}}

\affil[3]{\orgdiv{Center of Materials Science and Optoelectronics Engineering,}, \orgname{ University of Chinese Academy of Sciences}, \orgaddress{\city{Beijing}, \postcode{100049}, \country{China}}}

\affil[4]{\orgdiv{State Key Laboratory of Chemistry for NBC Hazards Protection}, \orgname{ PLA AMS}, \orgaddress{\city{Beijing}, \postcode{102205}, \country{China}}}

\affil[5]{\orgdiv{College of Science}, \orgname{National University of Defense Technology}, \orgaddress{\city{Changsha}, \postcode{410073}, \state{Hunan}, \country{China}}}

%\affil[2]{\orgdiv{Department}, \orgname{Organization}, \orgaddress{\street{Street}, \city{City}, \postcode{10587}, \state{State}, \country{Country}}}


\abstract{Ghost imaging~(GI) has demonstrated diverse imaging capabilities arising from its encoding-decoding-based computational imaging mechanism. 
Accordingly, information-theoretic studies have emerged as a promising direction for characterizing the fundamental performance limits of GI and related computational imaging paradigms.
Within the information-theoretic framework, the fundamental question of information-optimal encoding design in GI remains largely unexplored, primarily due to the intractability of the prior probability density function (PDF) of an unknown scene.  
Here, by leveraging the ability of recursively estimating the PDF of the object to be imaged via Bayesian filtering, we establish an adaptive information-maximization encoding~(AIME) design framework.  
Based on the adaptively estimated posterior PDF from previously acquired measurements, the expected information gain of subsequent detections is evaluated and maximized to design the corresponding encoding patterns in a closed-loop manner.
Under this framework, the theoretical form of the information-optimal encoding under the energy-constrained regime is analytically derived.
Corresponding experimental results show that, GI systems employing information-optimal encoding achieve markedly improved imaging performance compared with conventional fixed point-to-point imaging without relying on additional heuristic regularization, particularly in low signal-to-noise ratio regimes.
Moreover, the proposed AIME strategy is compared with existing encoding strategies under a representative experimental constraint, which consistently shows that it enables significantly enhanced information acquisition capability compared with existing encoding strategies, leading to substantially improved imaging quality. 
These results establish a principled information-theoretic foundation for optimal encoding design in computational imaging paradigms, provided that the forward model can be accurately characterized.}

\maketitle
\section{Introduction}\label{sec1}

Ghost imaging (GI)~\cite{moreau2018ghost} has demonstrated a broad range of capabilities, including acquisition of high-dimensional~(H-D) information through low-dimensional detection~\cite{eshun20253d}, enhancement of two-dimensional resolution via H-D light-field information~\cite{tong2020breaking,tong2025single}, non-locally decoupling detection efficiency and imaging resolution~\cite{yuan2021single}, and reduction of the unnecessary detection redundancy~\cite{hu2019optimization}.  
These advantages have enabled its diverse applications in imaging LiDAR~\cite{deng2017performance,pan2021micro}, single-shot H-D imaging cameras~\cite{liu2020spectral,chu2021spectral}, super-resolution microscopy~\cite{li2019single,chen2025multicolor}, low-dose and multi-dimensional radiographic imaging~\cite{zhang2018tabletop, he2021single}, imaging in special band~\cite{stantchev2020real,bogdanov2025ghost}, and task-oriented imaging~\cite{wu2025image,abbas2025target}.
From the perspective of imaging information acquisition~\cite{barrett2013foundations,shensheng2022review},  
% GI systems behave more in line with the encoding-decoding communication system than traditional imaging systems. 
GI can be interpreted as an encoding-decoding-based computational imaging paradigm~\cite{hu2022ghost}, where the imaging object is encoded by randomly fluctuating light fields, and the desired image is decoded through high-order light-field correlation. 
This formulation naturally places GI within an information-theoretic framework, under which analyses of information acquisition efficiency and fundamental performance are in principle attainable~\cite{li2017negative,hu2021correspondence,zhou2024theoretical}.

%And so far, several studies have discussed the imaging performance of GI with information-theoretic tools~\cite{enrong2013mutual,li2017negative,hu2021correspondence,zhentao2021some}.
Within the information-theoretic framework, a fundamental question in GI, and more broadly in such computational imaging paradigms, concerns the existence and design of information-theoretically optimal encoding~\cite{wei2016mutual}.
In information theory, the encoding optimality is commonly evaluated using information-content measures, such as mutual information~\cite{blahut1987principles} and Fisher information~\cite{fisher1922mathematical}. 
However, the probability density function~(PDF) of the object to be imaged, which is required for calculating these information measures, is generally intractable in practical imaging systems.
As a result, information-theoretic studies on encoding design of GI have so far been confined to very restrictive assumptions or special scenarios~\cite{enrong2013mutual} rather than a general imaging scenario.
As a practical compromise, encoding optimization has been extensively explored within the framework of compressed sensing~(CS), motivated by the principle that incoherent sampling can preserve the information of sparse signals with high efficiency~\cite{eldar2012compressed}. 
Assuming sparsity of the image in a specific transform domain, the encoding patterns in GI can be optimized by minimizing mutual coherence of the sensing matrix~\cite{xu2015optimization,czajkowski2018single,hu2019optimization}.
Nevertheless, such CS-based framework typically neglect the influence of measurement noise, as exact recovery guarantees are theoretically established only in the noiseless case~\cite{tropp2004greed}. 
In practical scenarios, the detection noise is often amplified~\cite{arias2011noise} in related image retrieval, rendering signal recovery particularly challenging in the low signal-to-noise ratio (SNR) regime.
Thus, existing approaches along this line remain theoretically incomplete, and are difficult to apply to GI under realistic low-SNR conditions.

Recently, by modeling the linear computational imaging system within the Bayesian filtering framework, Du \emph{et al.}~\cite{du2025information} proposed a recursive approach to estimate the PDF of the imaging object from already-acquired detection signals.
This formulation suggests a promising route for evaluating information-theoretic measures and designing information-optimal encoding light fields in an adaptive manner. 
However, from existing studies on adaptive GI~\cite{wu2020online,wang2023dual}, it remains unclear whether adaptive measurements can consistently yield fundamental performance gains, particularly from an information-theoretic perspective.
More broadly, in the context of adaptive CS, the superiority of adaptive sensing over non-adaptive measurements has long remained theoretically inconclusive~\cite{donoho2006compressed,arias2012fundamental,malloy2014near}, especially in noisy and practically constrained regimes. 

In this Article, we explicitly investigate the possibility to adaptively establish the information-optimal encoding for GI, and propose a general adaptive information-maximization encoding~(AIME) framework. 
It enables quantitative evaluation of the information content acquired in the encoding-detection process, and facilitates the design of optimal encoding under practical constraints through information maximization. 
Theoretical analysis and experimental results consistently reveal that the form  of information-optimal encoding is intrinsically governed by the underlying physical 
constraints. 
Under an ideal total-energy constraint corresponding to fixed illumination energy, the information-maximization encoding converges to a point-like illumination pattern, leading to an adaptive scanning strategy whose order is determined by the 
evolving object statistics rather than by a predetermined sequence. 
Owing to this adaptive design, AIME achieves superior imaging performance compared to conventional ``fixed point-to-point'' imaging modes without relying on additional heuristic regularization, particularly in low SNR regimes. 
Under realistic modulation constraints imposed by experimental hardware, although the optimal encoding no longer reduces to a point-like structure, it still emerges as an object-dependent adaptive pattern. 
In this regime, comparisons with representative encoding strategies show that AIME consistently provides improved imaging quality across a wide range of sampling ratio (SR) and SNR conditions. 
Further evaluation of the acquired information of different encoding strategies shows strong agreement with image quality, indicating that the performance gain originates from more efficient information acquisition rather than heuristic estimation effects. 
The proposed framework establishes a principled information-theoretic foundation for characterizing and approaching the fundamental performance of computational imaging paradigms.
%These results highlight the potential of information-theoretic adaptive encoding as a principled approach for performance enhancement in ghost imaging and more general computational imaging systems.
%\textit{Through comparisons with other types of encoding, we demonstrate that the proposed framework consistently out under different SNRs is demonstrated and verified in both the acquired information and the imaging quality. Moreover, it is observed that, by applying the information-optimal encoding to the under-sampling~(less than~$50\%$ sampling ratio~(SR)) case, the obtained imaging quality can achieve as good as that in the traditional imaging~($100\%$ SR) under the same SNR condition~(i.e., the illumination energy). }
% Also, the optimal design can combine with different kinds of regularization priors, like the sparsity prior obtained by learning. 
%Besides, as an adaptive encoding design mode, the fundamental limit is also investigated, which confirms the assertion of Cand{\'e}s that the adaptive sensing manner cannot ultimately improve the signal retrieval compared to the non-adaptive sensing manner in the sense of ensemble.
\FloatBarrier
\section{Results}\label{sec2}
\subsection{Principle of AIME for GI}
We formulate AIME as a closed-loop information-theoretic framework for GI, in which the encoding patterns are treated as variables to be optimized according to the evolving statistical state of the estimated imaging object. As illustrated in Fig.~\ref{fig:principle}~\textbf{b}, at each measurement step, the statistical state is represented by a posterior PDF inferred from previously detected signals, based on which the expected information gain of a candidate encoding is evaluated. The subsequent encoding is then determined by maximizing the information criterion under explicit physical followed by a recursive update of the posterior PDF after measurement.
Within this formulation, encoding, detection, and statistical inference in GI are integrated into a unified framework, enabling information-optimal encoding to be defined and implemented sequentially without requiring special assumption on prior knowledge of the object statistics.
%Figure~\ref{fig:principle} shows the schematic flowchart of the AIME framework in GI. 
%Based on the posterior PDF of the object to be imaged recursively estimated from already-acquired detection signals, the information contained in the subsequent detection signal can be evaluated by incorporating the forward detection model. 
%Then, the subsequent optimal encoding is adaptively optimized by maximizing the information content. And the object's image can be gradually estimated by updating the posterior PDF.  

%As shown in Fig.~\ref{fig:1}(a), in a typical GI system, the image formation process consist of several steps. First, the imaging object is illuminated by a set of modulated light fields generated from a digital micro-mirror device~(DMD). Then, the light reflected from the object is collected by the bucket detector.

\subsubsection{Bayesian formulation for recursive posterior PDF tracking} \label{sec:2.1.1}
To enable information-theoretic evaluation and optimization of encoding, the statistical state of the imaging object is characterized by a posterior PDF that is updated recursively as measurements are acquired.
To this end, the forward encoding-detection process of GI is formulated to construct the measurement-likelihood function~$p({\bf z} \mid {\bf x})$, specifically as~\cite{hu2022ghost, liu2025comprehensive}
\begin{equation}
{\bf z} = {\beta \bf H}{\bf x} + {\bf n}, \label{eq:z=Hx}
\end{equation} 
where~${\bf z}$ stands for the recorded signal, ${\bf x}$ denotes the desired object's image, ${\bf H}$ is the measurement matrix consisting of the encoding light fields~$\left\{ {\bf h}_1, {\bf h}_2, \cdots, {\bf h}_K \right\}^\top$ with~$K$ being the number of acquired measurements, $\bf n$ is the detection noise, and~$\beta$ represents a coefficient accounting for various system-dependent factors. This formulation is also shared by a broad class of computational imaging systems.

Then, an initial Gaussian distribution~$p({\bf x} \mid {\bf Z}_0) = p_0({\bf x})={\mathcal N}(\widehat{\bf x}_0, \widehat{\bf P}_0)$ is assumed by incorporating statistics of imaging scenes in the frequency domain~\cite{torralba2003statistics}(see Methods and Materials for more details). 
By considering the forward sub-model with a frame of light fields~${\bf h}_k$ constructed as~$p({\bf z}_k \mid {\bf x}) = {\mathcal N}(\beta{\bf h}_k^\top {\bf x}, R_k)$,
the posterior PDF after the~$k$-th measurement is obtained as
\begin{equation}
p({\bf x} \mid {\bf Z}_k) = \mathcal{N} (\widehat{\bf x}_k, \widehat{\bf P}_k).
\end{equation}
Here a Gaussian approximation~$n_k \sim {\mathcal N}(0,R_k)$ for the noise model is adopted, where~$R_k = \omega^2 {\bf h}_k^\top \widehat{\bf x}_{k-1}$ results from the photon noise assumption~\cite{liu2025comprehensive}, with the current estimate~$\widehat{\bf x}_{k-1}$ used to approximate the unknown true~${\bf x}$.
%photon noise assumption so that~$z_k \sim {\rm Poisson} \left(\beta {\bf h}_k^\top {\bf x}\right)$, and further approximate it as a Gaussian one such that~$R_k = \beta {\bf h}_k^\top {\bf x}$. 
%Since the true~${\bf x}$ is unknown during the imaging process, it is approximated with the current estimation~$\widehat{\bf x}_{k-1}$. As it can be seen below, the effectiveness of this approximation is empirically verified.
Under these settings, the posterior statistics~${\widehat {\bf x}}_k$, $\widehat{\bf P}_k$ are obtained via recursive expressions
\begin{subequations}
\begin{equation}
	{{\widehat {\bf x}}_k} = {\widehat {\bf x}}_{k-1}  + {{\bf K}_k}({{\bf z}_k} - \beta {{\bf h}_k^\top}{\widehat {\bf x}_{k-1}} ),
\end{equation}
\begin{equation}
	\widehat{\bf P}_k = ({\bf I} -\beta {{\bf K}_k}{{\bf h}_k^\top})\widehat{\bf P}_{k-1},
\end{equation}
\end{subequations}
with
\[
{{\bf K}_k} = \beta \widehat{\bf P}_{k-1} {\bf h}_k{( \beta^2 {{\bf h}_k^\top}\widehat{\bf P}_{k-1} {\bf h}_k + {R_k})^{ - 1}}.
\]
These posterior statistics serve as sufficient descriptors of the statistical state and provide the basis for quantitatively assessing the information content of subsequent measurements, as will be used in the following subsection.
\begin{figure}[t]
\centering
\includegraphics[width=0.73\linewidth]{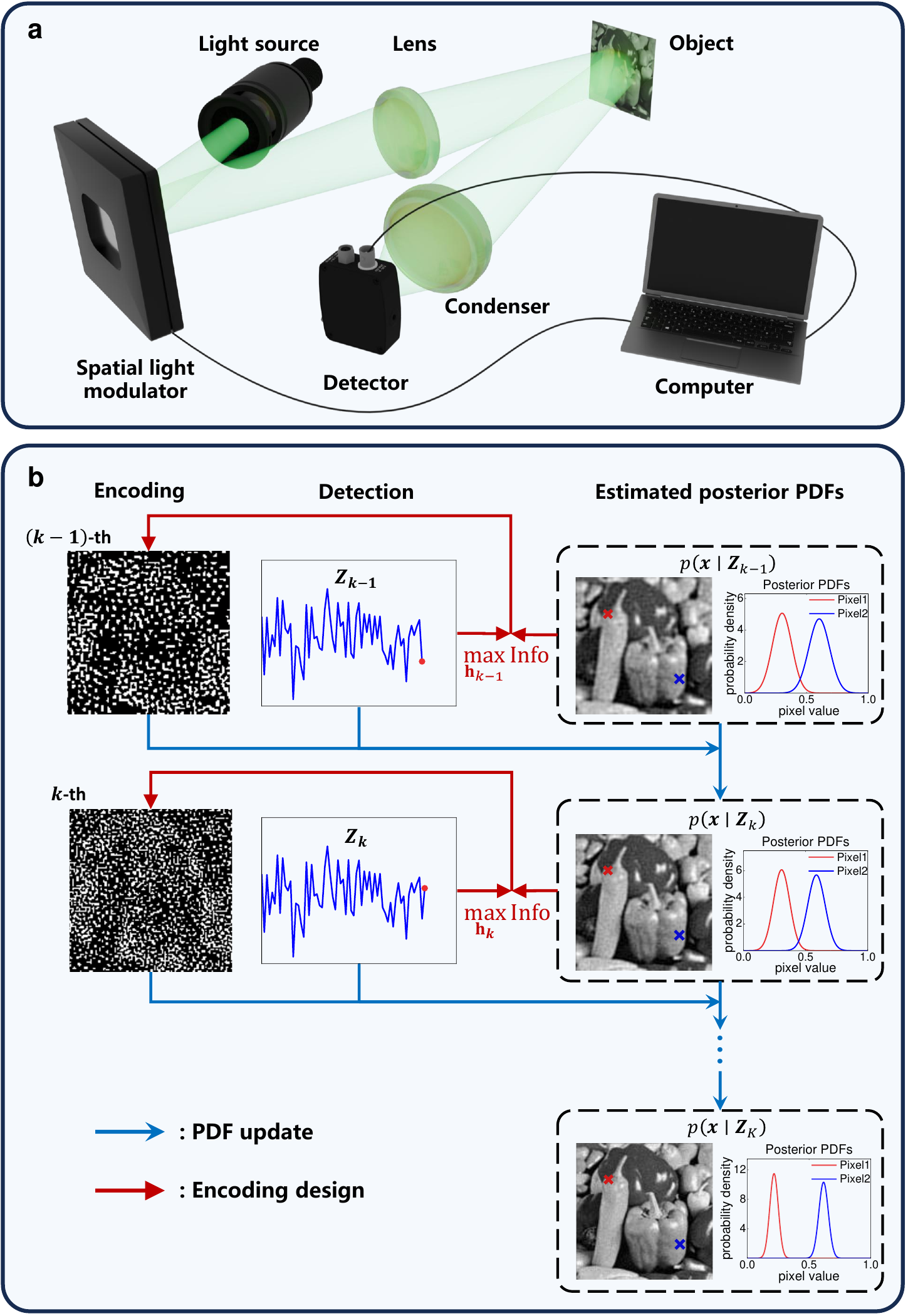}
\caption{\textbf{a.} Schematic illustration of the ghost imaging~(GI) system based on a spatial light modulator. \textbf{b.} Overview of the proposed adaptive information-maximization encoding~(AIME) framework for GI. The optimal encoding is designed by maximizing the information contained in the detection signal of the object to be imaged. The information is calculated based on the forward detection model and posterior PDFs.}
\label{fig:principle}
\end{figure}
\FloatBarrier

\subsubsection{Theoretical objective of AIME}
Based on the posterior PDF $p({\bf x} \mid {\bf Z}_k)$ obtained from recursive inference, the encoding for the $(k+1)$-th measurement can be formulated as a variable to be optimized according to an information-theoretic criterion. In the AIME framework, this optimization is guided by the expected reduction of posterior uncertainty induced by a candidate encoding, which provides a principled definition of information-optimal encoding. Here, the mutual information and the Fisher-information-related criterion are employed as two representative information-theoretic measures of this principle.

%Based on the posterior PDF~$p({\bf x} \mid {\bf Z}_k)$, the~$(k+1)$-th encoding can be designed by maximizing the information contained in the subsequent detection signal about the imaging object. In this study, the mutual information and Fisher information are used to serve as principled measures, specifically as follows.
\begin{itemize}
\item \textbf{The mutual information} between the~$k$-th detection signal and the imaging object~${\rm I}(z_{k}, {\bf x} \mid {\bf Z}_{k-1})$ is calculated by
\begin{equation}
	{\rm I}({\bf z}_{k}, {\bf x} \mid {\bf Z}_{k-1}) = \int {\rm d} {\bf z}_k {\rm d}{\bf x} p({\bf z}_k,{\bf x} \mid {\bf Z}_{k-1}) ) \log \frac{p({\bf z}_k,{\bf x} \mid {\bf Z}_{k-1})}{p({\bf z}_k) p({\bf x} \mid {\bf Z}_{k-1})}
\end{equation}
and obtained as
\begin{equation} 
	{\rm I}({\bf z}_{k}, {\bf x} \mid {\bf Z}_{k-1}) = \frac{1}{2} \log \frac{1}{\left|{\bf I} - \beta {{\bf K}_k}{{\bf h}_k^\top}\right|} = \frac{1}{2}\log \left(1+\beta^2 \frac{{{{\bf{h}}_k^\top}{{\widehat {\bf{P}}}_{k - 1}}{\bf{h}}_k}}{{{R_k}}} \right).
\end{equation}
Hence, maximizing the mutual information is equivalent to maximizing the following objective
\begin{equation} \label{eq:obj-MI}
	{\mathcal L}_{\rm MI} = \frac{{{{\bf{h}}_k^\top}{{\widehat {\bf{P}}}_{k - 1}}{\bf{h}}_k}}{{{R_k}}}.
\end{equation}
For convenience, this mutual-information-based formulation is hereafter referred to as AIME-MI. 

\item \textbf{The Fisher information} contained in the detection signal~$z_k$ is in fact represented by the Fisher information matrix~(FIM)~${\widehat{\bf P}}_k^{-1}$~\cite{du2025information}. Since the calculation of matrix inverse is time-consuming, instead of FIM, here the Cram{\'e}r-Rao bound~(CRB) is applied to perform the information-optimal encoding design by minimizing it. Specifically, the mean CRB~(mCRB) which stand for the error bound averaging on all pixels of the imaging object is calculated by 
\begin{equation}
	{\rm mCRB} = {\rm Tr} ({\widehat{\bf P}}_k)
\end{equation}
and obtain as
\begin{equation} \label{eq: mCRB expression}
	{\rm mCRB} = {\rm Tr} ({\widehat{\bf P}}_{k-1}) - \frac{{\bf h}_k^{\top} \widehat{{\bf P}}_{k-1}^2 {\bf h}_{k}}{{\bf h}_k^{\top} \widehat{{\bf P}}_{k-1} {\bf h}_{k}+\frac{R_k}{\beta^2}}.
\end{equation}
Since the covariance matrix~$\widehat{\bf{P}}_{k-1}$ has been determined by already-acquired detection, to minimize the mean CRB, only the last term in Eq.~(\ref{eq: mCRB expression}) can be optimized. Hence, minimizing the mCRB is equivalent to maximizing the following objective
\begin{equation} \label{eq:obj-CRB}
	{\mathcal L}_{\rm CRB} = \frac{{\bf h}_k^{\top} \widehat{{\bf P}}_{k-1}^2 {\bf h}_{k}}{{\bf h}_k^{\top} \widehat{{\bf P}}_{k-1} {\bf h}_{k}+\frac{R_k}{\beta^2}}.
\end{equation}
For convenience, this formulation is hereafter referred to as AIME-CRB.
\end{itemize}
With Eqs.~(\ref{eq:obj-MI}) and~(\ref{eq:obj-CRB}), it is possible to optimize the subsequent encoding~${\bf h}_k$ by
%To this end, we propose to maximize the mutual information~${\rm I}(z_{k}, {\bf x} \mid {\bf Z}_{k-1})$. 
\begin{equation} \label{eq:ObjFunc}
\max_{{\bf h}_k} {\mathcal L}, ~~~{\rm s.t.}~~ {\mathcal C}({\bf h}_k)
\end{equation}
where~${\mathcal C}({\bf h}_k)$ denotes the practical constraints on the encoding~${\bf H}_k$. 
We would like to note that, although derived from an estimation-theoretic perspective, the CRB-based criterion quantifies the same objective of posterior uncertainty reduction as the mutual-information-based formulation.

%These two formulations should be understood as different information-theoretic instantiations of the same AIME framework, rather than distinct encoding strategies, as both are derived from the posterior statistics and aim at optimizing the same information-driven objective.

%It is worth emphasizing that the proposed AIME framework is not restricted to a specific form of constraint or imaging system. As long as the forward model and the associated hardware constraints of the imaging system are accurately characterized, the AIME framework can be readily adapted to accommodate different practical limitations and system configurations. 
Under this formulation, the design of information-optimal encoding is unified as a constrained optimization problem driven by posterior PDF. Once a specific information criterion and the physical constraints on the encoding are specified, the structure of the optimal encoding follows directly from the solution of this optimization problem. Consequently, the AIME framework itself is not tied to a particular form of information measure or constraint, but provides a general principle for encoding design that can be instantiated under different system configurations.

%It is worth emphasizing that the AIME framework itself is not tied to a specific form of information criterion or constraint. Once the forward model and the associated physical constraints of an imaging system are characterized, the same posterior-driven optimization principle can be directly applied to accommodate different system configurations.

\FloatBarrier
\subsection{Information-optimal encoding in an ideal energy-constrained regime}
% \begin{figure}[t]
%	\centering
%	\includegraphics[width=0.99\linewidth]{Exper-Result-Point-to-Point_v260114-1002.pdf}
%	\caption{Illustration of the theoretical result of AIME patterns, each pattern is normalized to its maximum value for visualization. Figures \textbf{a1, a2} and \textbf{b1, b2} show the AIME results of two different imaging objects. \textbf{a1} and \textbf{b1} are patterns obtained by analytical solution of Eq.~(\ref{eq:AIME_MI}), each of which show exactly a scanning point; and the position of the point is marked with a purple arrow and box. \textbf{a2} and \textbf{b2} are patterns obtained by minimizing the mean CRB, which gradually become a small-size spot as the number of measurements increases; and the spot region in each pattern is displayed in an enlarged view. The first Hermite-Gaussian pattern is employed as the first frame of encoding pattern here.}
%	\label{fig:theo-result_v}
%\end{figure}
\begin{figure}[t]
	\centering
	\includegraphics[width=0.8\linewidth]{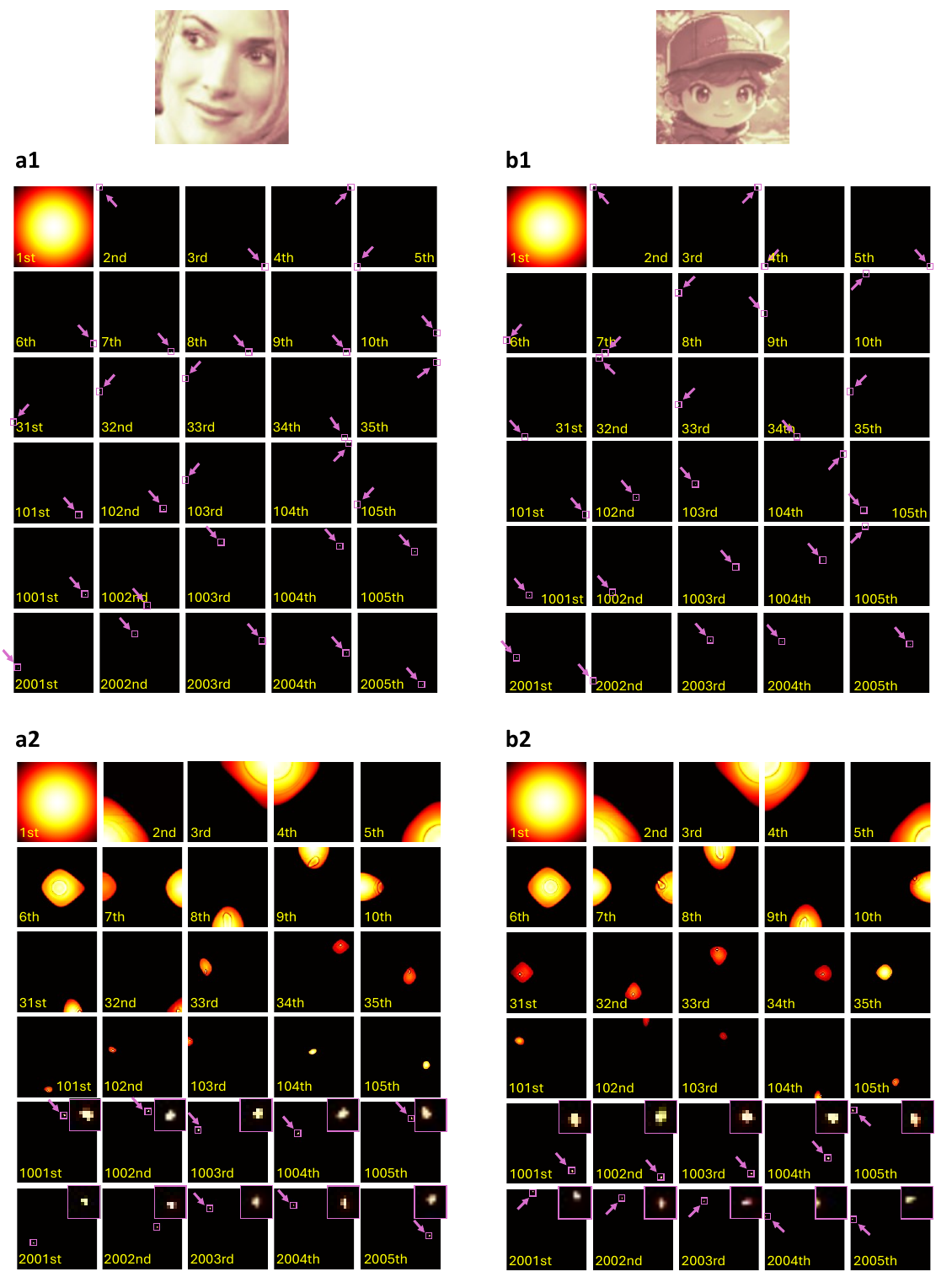}
	\caption{Illustration of the theoretical result of AIME-$\ell_1$ patterns, each pattern is normalized to its maximum value for visualization. Figures \textbf{a1, a2} and \textbf{b1, b2} respectively show AIME results of two different imaging objects. \textbf{a1} and \textbf{b1} are patterns obtained by analytical solution of Eq.~(\ref{eq:AIME_MI}), each of which show exactly a scanning point; and the position of the point is marked with a purple arrow and box. \textbf{a2} and \textbf{b2} are patterns obtained by minimizing the mean CRB, which gradually become a small-size spot as the number of measurements increases; and the spot region in each pattern is displayed in an enlarged view. The first Hermite Gaussian pattern is employed as the first frame of encoding pattern here for all cases.}
	\label{fig:theo-result_v}
\end{figure}
Based on the information-theoretic objective derived in the previous subsection, we first examine AIME in an idealized yet physically meaningful encoding regime. Specifically, we consider a total-energy constraint on the encoding light fields, which captures a fundamental limitation in optical imaging systems. Since the total energy of encoding~${\bf h}_k$ can be represented by its~$\ell_1$ norm, i.e., $\sum_j {{\bf h}_{k}}_j = \|{\bf h}_k\|_1 $, Aime under this constraint is denoted with suffix ``-$\ell_1$''.
%This regime serves as a benchmark scenario for revealing the intrinsic structure of information-optimal encoding.
In this regime, the optimal encoding admits a closed-form analytical solution, as established in the following theorem. 
%More practical constraints arising from hardware limitations and implementation considerations will be addressed in subsequent sections.

%Based on the information-theoretic objective derived in the previous subsection, we first present theoretical results of the AIME under a physically meaningful constraint on the encoding process. 
%Specifically, we impose a total-energy constraint on the encoding light field, which reflects a common practical limitation in optical imaging systems, allowing for a fair comparison of the fundamental performance among different encoding strategies.
%Under this constraint, the optimal encoding admits a closed-form analytical solution, as formalized in the following theorem. Other practical constraints arising from hardware limitations and implementation considerations will be examined later in the experimental section.
\begin{theorem}[Theoretical solution of information-optimal encoding under a total-energy constraint]\label{thm1}
	Consider the following optimization problem:
	\begin{equation}
		\begin{aligned}
			\max _{\mathbf{h}_k} & \frac{\mathbf{h}_k^{\top} \mathbf{P}_{k-1} \mathbf{h}_k}{\mathbf{h}_k^{\top} \mathbf{x}}, \\
			\rm { s.t. } & \left\|\mathbf{h}_k\right\|_1=C, \mathbf{h}_k \succeq \mathbf{0},
		\end{aligned}
		\label{eq:AIME}
	\end{equation}
	where~$\mathbf{P}_{k-1}$ is a positive definite covariance matrix, $\mathbf{x}$ is a non-negative vector, and~$C>0$ is a constant.
	Then the above problem admits an analytical optimal solution. Specifically, the optimal encoding vector is given by
	\begin{equation}
		\hat{\mathbf{h}}_k=\mathbf{e}_{\hat{i}}, \quad \hat{i}=\arg \max _i \frac{\left(\mathbf{P}_{k-1}\right)_{i i}}{x_i},
		\label{eq:AIME_MI}
	\end{equation}
	where~$\mathbf{e}_{\hat{i}}$ denotes the standard basis vector whose $\hat{i}$-th element is one and all other elements are zero.
\end{theorem}

%As a direct consequence of Theorem~\ref{thm1}, the information-optimal encoding that maximizes the mutual information admits a point-wise structure under the total-energy constraint.
%\begin{equation}
%	\widehat{\bf h}_k = \arg\max_i \frac{\left(\widehat{\bf P}_k\right)_{ii}}{\left(\widehat{\bf x}_k\right)_i}.
%\end{equation}
%This establishes that, in the ideal energy-constrained regime, information-optimal encoding necessarily reduces to an adaptive point-wise scanning strategy.
%
%
%As a direct consequence of~Theorem~\ref{thm1}, the information-optimal encoding of Eq.~(\ref{eq:obj-MI}) reduces to
%\begin{equation}
%	\widehat{\bf h}_k = \arg\max_i \frac{\left(\widehat{\bf P}_k\right)_{ii}}{\left(\widehat{\bf x}_k\right)_i},
%	\label{eq:AIME_MI}
%\end{equation}
%where~$\left(\widehat{\bf P}_k\right)_{ii}$ denotes the~$i$-th diagonal element of matrix~$\widehat{\bf P}_k$ and~$\left(\widehat{\bf x}_k\right)_i$ is the~$i$-th entry of~$\widehat{\bf x}_k$. 
This Theorem establishes that, in the ideal energy-constrained regime, information-optimal encoding necessarily reduces to an adaptive point-wise scanning strategy. 
Importantly, this point-wise structure does not imply a predetermined or fixed scanning order, but instead, varies according to the evolving posterior statistics.
For the CRB-based objective shown as Eq.~(\ref{eq:obj-CRB}), although a closed-form analytical solution of the optimal encoding  is difficult to derive, simulation and experimental results indicate a consistent structural evolution behavior: the optimal encoding initially appears as a diffuse illumination pattern and gradually concentrates into a localized spot as the number of acquired measurements increases.
\begin{figure}[t]
	\centering
	\includegraphics[width=0.9\linewidth]{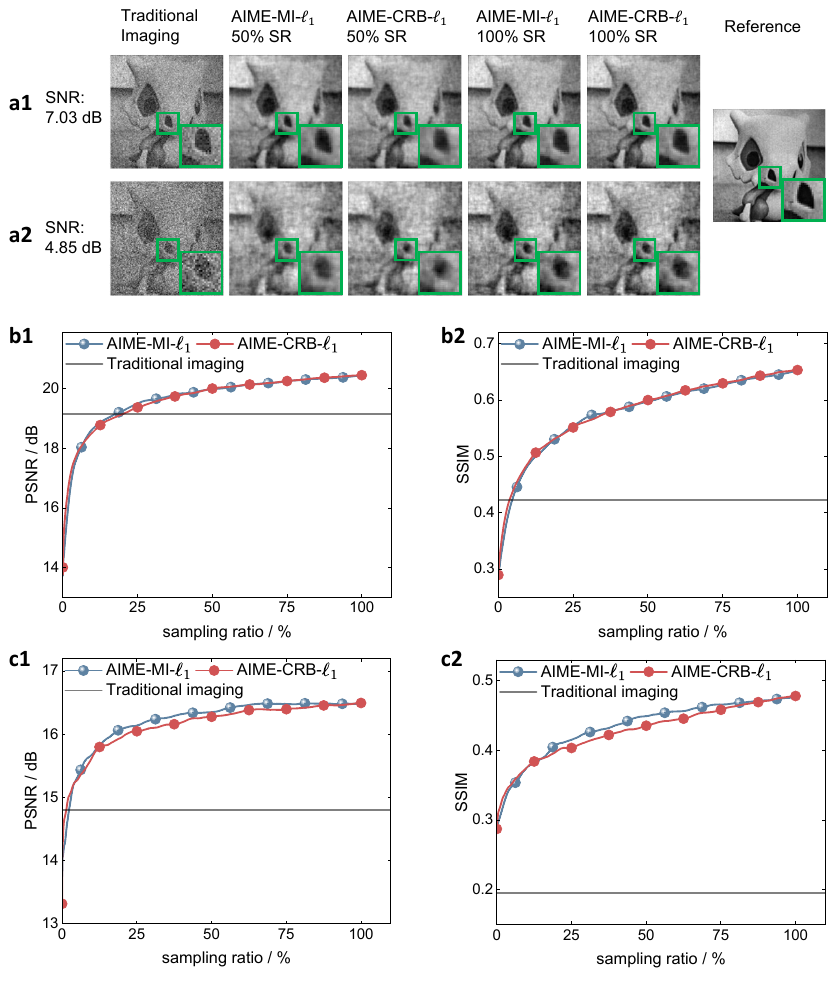}
	\caption{Experimental comparison of imaging results of GI with AIME-$\ell_1$ and conventional imaging. \textbf{a:} Representative imaging results obtained by GI AIME-$\ell_1$ and conventional imaging under sampling ratios (SRs) of 50\% and 100\%. Results are shown for two signal-to-noise ratio (SNR) levels, as indicated on the left. The reference image  is shown on the right, which is obtained by full-sampling point-wise scanning under a significantly high SNR.
		\textbf{b:} Quantitative evaluation of imaging quality in terms of PSNR and SSIM as functions of the sampling ratio for the case of~${\rm SNR} = 7.03~{\rm dB}$. The horizontal gray lines indicate the performance of conventional full-sampling (100\% SR) imaging for comparison.
		\textbf{c:} Corresponding PSNR and SSIM curves for the case of~${\rm SNR} = 4.85~{\rm dB}$. The horizontal gray lines again indicate the performance of conventional full-sampling (100\% SR) imaging for comparison.}
	\label{fig:exper-result-2}
\end{figure}
Figure~\ref{fig:theo-result_v} provides an illustrative visualization of the AIME-$\ell_1$ encoding patterns, where panels~\textbf{a} and~\textbf{b} correspond to two representative imaging objects. To ensure a fair comparison, the first encoding pattern is identically initialized as the first Hermite Gaussian mode~\cite{huang2024compressed} for all cases. The results visually corroborate the theoretical prediction that information-optimal encoding evolves toward an adaptive point-wise structure under the total-energy constraint, which fundamentally differs from conventional raster-scan imaging in that the scanning order is not fixed but adaptively determined by the statistical structure of the object.
This difference between AIME-$\ell_1$ and conventional imaging lead to an improved imaging performance. To show that,
A quantitative experimental comparison under identical total energy and noise conditions is performed. 
As shown in Fig.~\ref{fig:exper-result-2}, the experimental results indicate that AIME consistently achieves higher imaging quality than traditional imaging, particularly in low-SNR regimes. 
In specific, with approximately 50\% sampling, AIME yields a visually recognizable imaging result at a detection SNR of 7.03 dB, whereas conventional imaging with 100\% sampling fails to recover the object reliably. When both methods operate at full sampling, AIME produces noticeably improved image quality, with an increase of about 0.2 in SSIM. 
As the detection SNR further decreases to 4.85 dB, the object becomes nearly indistinguishable in conventional imaging, while AIME remains capable of providing a coarse yet informative image. 
Details of the experimental setting for this comparison is presented in Materials and Methods, and more experimental comparison results on representative scenes are provided in the supplementary materials.
These results demonstrate that the AIME strategy enables more efficient information utilization than fixed raster-scan acquisition in the ideal energy-constrained regime.

An alternative intuition for the emergence of point-wise optimal encoding can be drawn from CS. It is well established in CS theory that sparse signals can be efficiently acquired through incoherent measurements in a conjugate domain, for example, spatially sparse signals can be recovered from partial Fourier measurements~\cite{candes2006robust}. In the present imaging scenario, natural scenes are often sparse in transform domains such as the discrete cosine transform~(DCT) domain. From this perspective, concentrating measurements on a limited set of spatial locations can be viewed as an intuitive counterpart to selectively probing sparse degrees of freedom in a conjugate domain, offering a qualitative analogy to the point-wise scanning behavior predicted by the information-theoretic analysis.

\FloatBarrier
\subsection{Information-optimal encoding under realistic modulation constraints}
The preceding analysis establishes the intrinsic structure of information-optimal encoding under an idealized total-energy constraint. In practical GI systems, however, additional modulation constraints on the encoding patterns are inevitably imposed. It is therefore essential to examine whether the advantages of AIME persist under such realistic conditions. To this end, beyond the ideal total-energy constraint, we further impose a bounded-amplitude constraint on the encoding patterns,
\begin{equation}
	0 \le h_{k,i} \le C, \quad \forall i ,
\end{equation}
which reflects realistic hardware limitations such as the finite modulation depth of spatial light modulators or digital micromirror devices. 
To differentiate with AIME-$\ell_1$, AIME under this constraint is denoted with suffix ``-$\ell_\infty$'', since the maximal element of non-negative encoding~${\bf h}_k$ can be represented by its~$\ell_\infty$ norm.

Under such constraints, the closed-form analytical solution derived above no longer applies, and the information-optimal encoding is therefore obtained numerically by solving the resulting constrained optimization problem. 
Figure~\ref{fig:illustrate adaptive} provides an illustrative visualization of the joint evolution of  imaging results and the optimized encoding patterns under the proposed strategy. 
Two different imaging scenes, namely \textit{cameraman} and \textit{pepper}, are considered, and the peak signal-to-noise ratio (PSNR) is used to quantify the imaging quality.
Figures~\ref{fig:illustrate adaptive}~\textbf{(a1)} and~\textbf{(b1)} show imaging results together with the corresponding optimized encoding patterns at different sampling ratios (SRs), where the selected SRs correspond to the marked points in Figs.~\ref{fig:illustrate adaptive}~\textbf{(a2)} and~\textbf{(b2)}, respectively. 
The optimized encoding patterns exhibit a clear coarse-to-fine evolution as the SR increases. At low SRs, the patterns consist of large-scale speckle-like structures, whereas at higher SRs, finer spatial features gradually emerge. 
This behavior is consistent with the gradual refinement of the posterior distribution in the Bayesian filtering process, where early measurements focus on coarse information content and subsequent measurements concentrate on increasingly localized spatial details.
The corresponding PSNR curves shown in Figs.~\ref{fig:illustrate adaptive}~\textbf{(a2)} and~\textbf{(b2)} demonstrate a rapid improvement in reconstruction quality at the early stage of sampling, indicating the high information efficiency of the proposed strategy under limited measurements. To further illustrate the accuracy and uncertainty evolution of the reconstruction, the posterior probability density functions (PDFs) of two representative pixel locations (marked by red and blue crosses in the estimated images) are plotted in Figs.~\ref{fig:illustrate adaptive}~\textbf{(a3)} and~\textbf{(b3)}. As SR increases, the posterior PDFs become progressively sharper, reflecting the systematic reduction of estimation uncertainty and the convergence of pixel-wise intensity estimates.
By comparing the results obtained for different imaging scenes, it is evident that the optimized encoding patterns are also strongly scene-dependent under the bounded-amplitude constraint. This observation is fully consistent with the theoretical formulation of AIME and highlights its ability to adaptively allocate measurements according to the estimated statistical state of the imaging object.

\begin{figure}[t]
	\centering
	\includegraphics[width=0.85\linewidth]{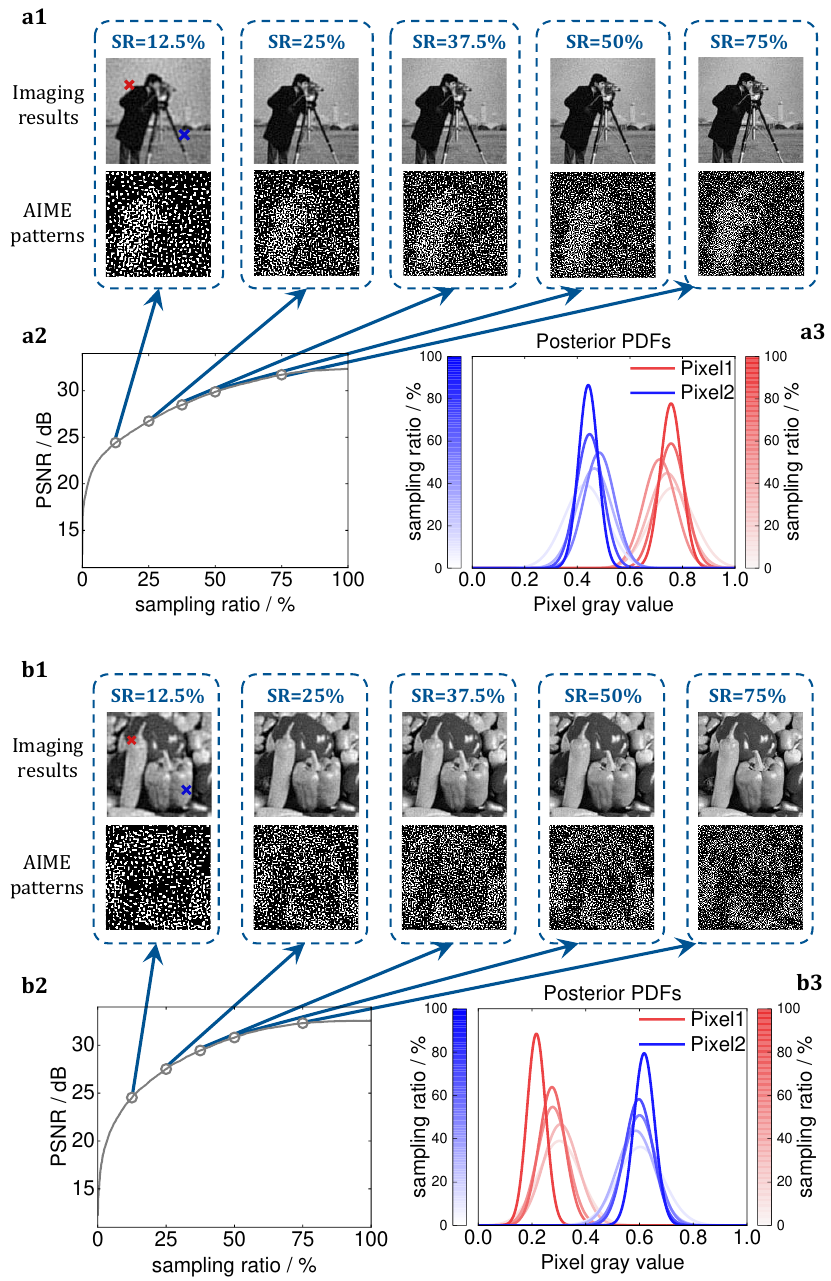}
	\caption{Simulation results for illustrating the evolution process of the proposed strategy on two different imaging scenes. \textbf{a1} and \textbf{b1} show the evolution of estimated images and corresponding optimized encoding patterns with sampling ratio increasing for two imaging scenes, respectively; \textbf{a2} and \textbf{b2} show evolution of PSNR indicators with the increase in sampling; \textbf{a3} and \textbf{b3} show evolution of estimated PDFs of two selected pixels as marked on the imaging results, respectively.}
	\label{fig:illustrate adaptive}
\end{figure}

\FloatBarrier

\begin{figure}[t]
	\centering
	\includegraphics[width=0.9\linewidth]{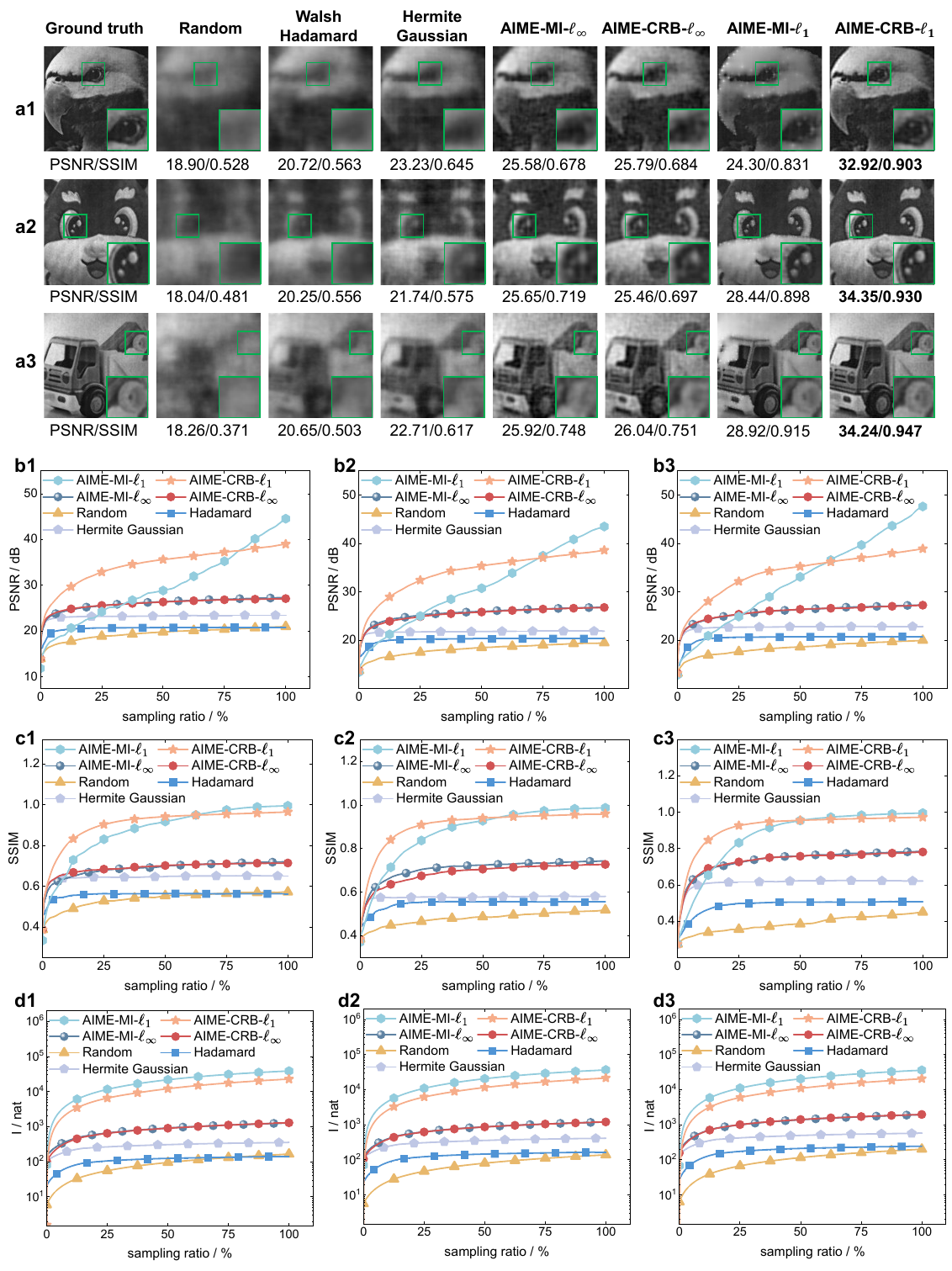}
	\caption{Experimental imaging results for comparison between AIME-$\ell_\infty$  and some typical encoding patterns. Results of AIME-$\ell_1$ that obtained under the total-energy constraint are also shown here to serve as a fundamental-limit reference. \textbf{a1}, \textbf{a2}, and \textbf{a3} show three different representative cases. The SR and SNR for them are: \textbf{a1} SNR around 18~dB, and 25~\% SR; \textbf{a2} SNR around 18.5~dB, and 37.5~\% SR; \textbf{a3} SNR around 20~dB, and 37.5~\% SR. PSNR and SSIM indicators of imaging results corresponding to three scenes as SR increases are shown in \textbf{b1, c1}; \textbf{b2, c2}, and \textbf{b3, c3}, respectively. To further quantify and confirm the information acquisition enhancement, the acquired information of different encoding strategies is calculated, and shown in \textbf{d1}, \textbf{d2}, and \textbf{d3}.}
	\label{fig:exper-result}
\end{figure}

To evaluate the practical effectiveness of AIME-$\ell_\infty$ under such realistic experimental constraints, we compare it with several representative existing encoding patterns in GI, including pseudo-thermal random patterns, Walsh-Hadamard patterns, and Hermite Gaussian patterns~\cite{huang2024compressed}. 
For theoretical reference, we also include the ideal total-energy-constrained AIME-$\ell_1$ results, which serves as a theoretical performance limit without experimental modulation constraints. 
All results are obtained by employing the same Bayesian filtering framework to ensure a fair comparison that isolates the effect of encoding design. 
Figure~\ref{fig:exper-result} summarizes the results of three representative scenes acquired under different sampling ratios (SRs) and signal-to-noise ratios (SNRs). 
The imaging results are shown as Figs.~\ref{fig:exper-result} \textbf{a1-a3}, respectively. 
For all tested scenes, both AIME-based strategies (i.e., AIME-MI and AIMe-CRB) consistently yield superior imaging quality compared with other encoding patterns under the same modulation constraints, as evidenced by clearer structural details and significantly reduced noise artifacts. 
The corresponding PSNR and SSIM values are reported under each image, which confirm that AIME achieves much better quantitative performance across all conditions. 
At the same time, a performance gap between AIME-$\ell_\infty$ and AIME-$\ell_1$ theoretical reference can be observed, suggesting the loss of achievable information imposed by hardware constraints on the encoding space. 
The evolution of PSNR and SSIM as functions of the SR is further plotted in Figs.~\ref{fig:exper-result} \textbf{b1}-\textbf{b3} and \textbf{c1}-\textbf{c3}, respectively. 
These results reveal that AIME-$\ell_\infty$ exhibits a pronounced advantage at early measurement stages, where substantial imaging quality improvements are already achieved at low SRs. 
As the number of measurements increases, AIME-$\ell_\infty$ continues to provide steady performance gains, whereas the imaging quality obtained with other encoding patterns rapidly saturates. 
This behavior demonstrates the strong compressed-sampling 
capability of AIME and its effectiveness in allocating the most informative measurements. 
The gap between AIME-$\ell_\infty$ and AIME-$\ell_1$ can be also reflected from imaging quality curves, which quantifies the performance loss induced by realistic modulation constraints.

%Notably, the reconstruction results indicate only marginal differences between the two information-optimal objectives, namely mutual-information-based and Fisher-information-based encoding. 
To further quantify the information acquisition capability of different encoding strategies, we explicitly compute the cumulative mutual information accumulated from obtained detection signals. Specifically, the cumulative mutual information is calculated as $\sum_k \frac{1}{2}\ln \left( 1 + \frac{{{\bf H}_k \widehat{\bf P}_{k-1} {\bf H}_k^\top}}{{R_k}} \right)$ and plotted on the logarithmic scale.
The corresponding results shown in Figs.~\ref{fig:exper-result} \textbf{d1-d3} represent evolution of the acquired amount of information about the unknown imaging scene during the detection process. 
It can be found that, across different SRs and noise levels, AIME consistently enables the detection signals to acquire substantially more information about the imaging object than other comparison encoding strategies. 
Also, the separation between AIME-$\ell_\infty$ and AIME-$\ell_1$ curves quantifies the information loss due to realistic modulation constraints, and mirrors the gap observed in reconstruction quality. 
Moreover, the ordering of the information curves agrees with the ordering of reconstruction quality, establishing a direct correspondence between information acquisition and imaging performance. 
This agreement confirms that the superiority of AIME arises from principled information maximization rather than other heuristic effects.

%\subsection{Analysis on the acquired information}
%%Since AIME takes detection SNR into account, it can well deal with different SNR conditions.  Fig
%\begin{figure}[t]
%	\centering
%	\includegraphics[width=0.95\linewidth]{Simulation-Result-Info-251204-1100.pdf}
%	\caption{\textbf{a1, a2} \textbf{a3} and \textbf{b1, b2 ,b3} show the cumulative mutual information and the cumulative Fisher information  under different SRs  contained in detection signals from different kinds of encoding patterns, respectively. \textbf{a1, b1}, \textbf{a2, b2} and \textbf{a3, b3} correspond to imaging scenes shown in Fig.~\ref{fig:exper-result} \textbf{a1, a2} and \textbf{a3}, respectively.
	%	}
%	\label{fig:simu-result-infoContent-different-SNRs}
%\end{figure}

%\subsection{Applying to other types of practical constraints}
%For example, if we consider that  encoding light fields are generated via DMD, the constraint 
%
%can be used, where the first term is limited by the reflectance of DMD, and the second term represents the fact that light field patterns are non-negative.
\FloatBarrier
\section{Discussion}\label{sec12}
In this work, an adaptive information-maximization encoding strategy is developed for GI. The proposed framework is intended to serve as a generally applicable approach for information-optimal encoding in computational imaging systems with the encoding-decoding mode, provided that the forward imaging model can be explicitly characterized.
In the present study, the forward model is formulated as a simplified linear model and demonstrated using a spatial-domain GI system with sequential bucket detection. 
While this simplified model enables analytical tractability and conceptual clarity, more realistic forward models accounting for optical blurring, down-sampling, and noise~\cite{huang2020ghost,huang2021ghost,liu2025comprehensive} can be further incorporated into the proposed framework. 
Moreover, since Bayesian filtering theory is not restricted to linear systems, the proposed strategy can, in principle, be generalized to computational imaging problems governed by nonlinear forward models, such as the coherent detection GI~\cite{pan2021micro}.
Beyond the form of the forward model, the framework is also not restricted to a specific sensing configuration. 
Although AIME is demonstrated with sequential bucket detection, the underlying Bayesian filtering and information optimization principles can be naturally extended to batch-measurement settings. 
In such cases, multiple measurements are acquired 
simultaneously, and the framework can be generalized to adaptive or semi-adaptive encoding design in a higher-dimensional measurement space. 
This suggests its applicability to a broader class of computational imaging systems with array detection, like the GI camera~\cite{chu2021spectral}.

In addition to the forward model, another essential component of the proposed framework is the setting of the initial image PDF. In the presented study, a general prior is adopted by imposing constraints in the image frequency domain, which allows an PDF to be defined for an entirely unknown imaging scene. 
While this choice provides a concise and robust prior for Bayesian filtering, it can be further improved by incorporating alternative analytical image prior that better capture higher-order statistical properties of natural images~\cite{hyvarinen2009natural}. 
Furthermore, with the development of deep neural models, image priors can now be effectively learned from data using deep generative models. 
In scenarios where a sufficiently large and representative training dataset is available, the proposed strategy can be naturally combined with image priors provided by trained generative models~\cite{elata2024adaptive,pinkard2024information}. 
Such learned priors are expected to offer a more expressive statistical representation than the general analytical prior employed in this work, potentially leading to further performance improvements.

Altogether, this work highlights the potential of information-theoretic encoding design as a principled approach to GI and, more broadly, to computational imaging systems. 
By explicitly modeling the information carried by detection signals under realistic physical constraints imposed by the imaging system, the proposed framework provides a systematic methodology to characterize and approach fundamental imaging limits.
The emergence of point-like encoding as information-optimal under a total-energy constraint offers an information-theoretic insight suggesting that point-to-point sampling strategies observed in biological vision~\cite{land2012animal} may reflect a principle of efficient information acquisition. 
Beyond GI, the proposed framework is expected to inspire further information-theoretic studies on the evaluation and optimization of imaging and sensing systems, with potential implications for a broad range of applications, including autonomous sensing, machine vision, and data-efficient perception tasks. 
Importantly, in numerous such systems, flexible light-field modulation platforms such as metasurfaces and programmable photonic devices are increasingly employed as core encoding components, which provides structured design space for encoding optimization. 
Incorporating this into the proposed framework is supposed to offer a principled pathway that integrates emerging hardware platforms with the information-optimal 
encoding design.

\section*{Materials and Methods}
\subsection*{Details for  encoding optimization }
The encoding optimization problem is formulated as a constrained optimization task and solved using Projected Gradient Descent (PGD), which alternates between gra-
dient descent updates and projection onto the feasible constraint set. 
During the gradient descent step, we employ the SGD optimizer with learning rates of~$1 \times 10^{9}$ and $1 \times 10^{5}$ for the mutual-information-based and CRB-based objectives, respectively. 
All remaining optimizer hyperparameters follow the default PyTorch configuration. 
To balance computational efficiency and estimation accuracy, each encoding optimization is performed using 100 PGD iterations. 
All numerical calculation are implemented in PyTorch 2.5.1 with CUDA 12.1 and executed on an NVIDIA GeForce RTX 3090 GPU. 
For a $128 \times 128$ image, one complete cycle consisting of Bayesian filtering and encoding optimization requires approximately 0.12 s.

%The optimization of the objective function is formulated as a constrained problem and solved using Projected Gradient Descent (PGD), which alternates between gradient descent and projection of parameters onto the constraint domain. Specifically, during the gradient descent step, we employ the SGD optimizer with learning rates of~$1 \times 10^{9}$ and $1 \times 10^{5}$ for the mutual information and CRLB objective function, respectively. All other hyperparameters are set to the default values in PyTorch. To balance computational efficiency and estimation quality, we perform 100 PGD iterations per encoding optimization. All simulations are implemented in PyTorch 2.5.1 with CUDA 12.1 and executed on an NVIDIA GeForce RTX 3090 GPU. One epoch comprising both Bayesian filtering and encoding optimization takes approximately 0.2~s for a~$128 \times 128$ image.

\subsection*{Measurement of parameters~$\beta$ and~$R_k$ in the experiment}
The experiment is performed with a GI system with DMD modulation. To determine 
the forward model in Eq.~(\ref{eq:z=Hx}), the scaling factor~$\beta$ and the noise statistics~$\{R_k\}$ are 
experimentally calibrated prior to imaging.

\paragraph{Calibration of~$\beta$}
A whiteboard with uniform reflectance is used as the calibration object. A set of~$M$ Hadamard encoding patterns~${\bf h}^{(1)}, \cdots, {\bf h}^{(M)}$ is projected sequentially, and for each pattern~$L$ repeated measurements are recorded, yielding detection vectors~${\bf z}^{(m)} \in {\mathcal R}^L$. Then, $\beta$ is estimated by
\begin{equation}
	\beta = \left\langle \frac{\overline{ z}^{(m)}}{\sum_i {\bf h}^{(m)}_i} \right\rangle_m
\end{equation}
where~$\overline{ z}^{(m)}$ denotes the average on elements of~${\bf z}^{(m)}$ and~$\left\langle \cdot \right\rangle_m$ represents averaging over all calibration encoding patterns.

\paragraph{Calibration of noise statistics}
As modeled in Section~\ref{sec:2.1.1}, the noise variance is parameterized as
\begin{equation*}R_k = \omega^2 {\bf h}_k^\top {\bf x}\end{equation*}
where~$\omega^2$ characterizes the proportionality between signal intensity and noise variance 
and therefore controls the effective detection SNR.
With those measurements on the whiteboard by using Hadamard patterns, estimation on~$\omega^2$ in practical conditions is as follows. For the~$m$-th encoding pattern, with~$L$ 
repeated measurements, the corresponding noise coefficient is estimated as
\begin{equation}
	{\omega^2}^{(m)} = \frac{{\sigma^2}^{(m)}}{{\bar z}^{(m)}}
\end{equation}
where ${\sigma^2}^{(m)} = \frac{\sum_{l=1}^{L} \left( z_l^{(m)} - \bar{z}^{(m)} \right)^2}{L}$ denotes the sample variance of the detection signal.
The calibrated noise parameter is then taken as the ensemble average
\begin{equation}
	\omega^2_{\text{eff}} = \left\langle{\omega^2}^{(m)} \right\rangle_m.
\end{equation}
By adjusting the illumination intensity and repeating the calibration procedure, noise 
parameters corresponding to different operating SNR regimes are obtained. In each regime, the calibrated~$\omega^2_{\text{eff}}$ is used consistently in both mutual-information-based and 
CRB-based encoding optimization.

\subsection*{Experimental setting in comparison between GI with AIME-$\ell_1$ and conventional imaging}
For experimental comparison between GI with AIME-$\ell_1$ and ``fixed point-to-point'' conventional imaging, both strategies are implemented using the same DMD-based GI system to ensure identical optical and detection conditions. 
For AIME-$\ell_1$, the optimized encoding patterns are physically realized with DMD. 
To improve illumination stability, each~$4\times4$ block of DMD micro-mirrors is grouped to form a single effective illumination pixel. 
The optimized patterns are mapped onto these grouped pixels and sequentially projected onto the scene, and the corresponding bucket detection signals are recorded. 
For conventional imaging, the same DMD-based GI system is operated in a raster-scanning mode. 
Specifically, a sequence of point-like illumination patterns is loaded onto the DMD, where each grouped~$4\times4$ block is illuminated one at a time following a fixed scanning order. 
The detected bucket signal associated with each illumination position is directly assigned as the grayscale value of the corresponding image pixel, thereby forming the image.

In the experiment under low-SNR condition, additional background noise is introduced using scattering light generated by a rotating ground glass. 
This produces a dynamic background whose temporal fluctuation, after bucket detection and time averaging, can be well approximated as additive Gaussian noise with constant variance.
To quantify the noise level, $L$ measurements are acquired using a completely dark pattern. Let $z_{\text{back}}^{(l)}$ denote the $l$-th recorded signal, and $\bar{z}_{\text{back}}$ its sample mean over~$L$ measurements, the background noise variance $\sigma_{\text{back}}^2$ is estimated as:
\begin{equation}
	\sigma_{\text{back}}^2 = \frac{\sum_{l=1}^{L} \left( z_{\text{back}}^{(l)} - \bar{z}_{\text{back}} \right)^2}{L - 1},
\end{equation}
With experimental detection signal~$\{z_k\}$ corresponding to both AIME and raster-scanning, the experimental SNR is estimated by
\begin{equation}
	\text{SNR} = 10 \log_{10} \left( \frac{{\bar z} - \bar{z}_{\text{back}}}{\sigma_{\text{back}}} \right).
\end{equation}
where~$\bar z$ is the average intensity of detection signals.

\subsection*{Setting the initial image PDF}
In the proposed strategy, an initial Gaussian PDF~$p({\bf x})={\mathcal N}({\bf x}_0, {\bf P}_0)$ of the object's image is set. To specify the parameters~${\bf x}_0$ and~${\bf P}_0$,we adopt a commonly used statistical model for natural images, in which the frequency-domain coefficients~$\left\{ a({\bf  f})  \right\}$ of images~\cite{torralba2003statistics} follow a Laplacian distribution
\begin{equation}
	p(a({\bf f})) = \frac{1}{2 \lambda({\bf f})} \exp \left[ - \frac{\left|a({\bf f}) - \mu ({\bf f})\right|}{\lambda({\bf f})} \right],
	\label{eq:lapla}
\end{equation}
where~$\mu ({\bf f})$ for those non-zero-frequency components is set as 0, $\mu({\bf f}_0) = \mu_0$ representing the zero-frequency amplitude is estimated with a uniform illumination, and  $\lambda({\bf f}) \propto (|{\bf f}|)^{-\gamma}$ with~$\gamma \in [1, 2]$.
It should be noted that, owing to the high dimensonality of image signals and the central limit theorem, although the frequency-domain coefficients are modeled as independent Laplacian variables, the image-domain pixels, being weighted superpositions of a large number of such coefficients, can be well approximated to be subject to a Gaussian distribution. 
Based on this, the parameters of the initial Gaussian PDF are set as
\begin{subequations}
	\begin{equation}
		{\bf x}_0 = \mu_0 {\bf 1} 
	\end{equation}
	\begin{equation}
		{\bf P}_0 = 2\Psi^\top  {\bf \Sigma}^2 \Psi
	\end{equation}
\end{subequations}
where~$\bf 1$ is an all-$1$ vector, ${\bf \Sigma}$ is a diagonal matrix whose diagonal elements consist of $\left\{\lambda({\bf f}_0), \lambda({\bf f}_1), \cdots \right\}$ and~$\bf \Psi$ denotes DCT matrix.

\backmatter

%\section*{Supplementary information}
%
%Materials include the Supplementary Text, Figs.~S1 to S10, and Reference~[1]--[3].

%Please refer to Journal-level guidance for any specific requirements.

%\bmhead{Acknowledgements}
%
%Acknowledgements are not compulsory. Where included they should be brief. Grant or contribution numbers may be acknowledged.
%
%Please refer to Journal-level guidance for any specific requirements.

\section*{Data availability}
Data are available from the corresponding author on request.

\section*{Funding}
This work is supported by the National Natural Science Foundation of China~(62201165, 62475270) and the National Key Research and Development Program of China~(2024YFF0505601).

\section*{Author contributions}
C. Hu and S. Han conceived the idea and supervised the study. J. Sun and C. Hu performed most of the theoretical derivation.  J. Sun and C. Hu conducted the experiment under discussion with Z. Bo, Z. Liu and M. Chen.  L. Du and W. Liu contribute to code modification for simulation and experimental data analyses.  J. Sun and C. Hu analyzed the experimental data, and wrote the original draft. All authors discussed the results and contributed to editing and revising the paper.

\section*{Competing interests}
The authors declare no competing interests.

%Some journals require declarations to be submitted in a standardised format. Please check the Instructions for Authors of the journal to which you are submitting to see if you need to complete this section. If yes, your manuscript must contain the following sections under the heading `Declarations':

%\bmhead{Funding} 
%
%\begin{itemize}
%\item Conflict of interest/Competing interests (check journal-specific guidelines for which heading to use)
%\item Ethics approval and consent to participate
%\item Consent for publication
%\item Data availability 
%\end{itemize}

%\noindent
%If any of the sections are not relevant to your manuscript, please include the heading and write `Not applicable' for that section. 

\bigskip
%\bigskip\noindent
%Springer journals and proceedings: \url{https://www.springer.com/gp/editorial-policies}
%
%\bigskip\noindent
%Nature Portfolio journals: \url{https://www.nature.com/nature-research/editorial-policies}
%
%\bigskip\noindent
%\textit{Scientific Reports}: \url{https://www.nature.com/srep/journal-policies/editorial-policies}
%
%\bigskip\noindent
%BMC journals: \url{https://www.biomedcentral.com/getpublished/editorial-policies}
%\end{flushleft}
\newpage
\begin{appendices}
\renewcommand{\thefigure}{A\arabic{figure}}

\section{Theoretical derivation of the objective function for AIME-MI and AIME-CRB}
In the proposed strategy, the mutual information~${\rm I}(z_{k}, {\bf x} \mid {\bf Z}_{k-1})$  is specifically calculated as follows
% \begin{equation}
	% 	\begin{array}{c}
		% 		I(\hat x_k^ - ,{z_k}) = H(\hat x_k^ - ) - H(\hat x_k^ - |{z_k})\\
		% 		= H(\hat x_k^ - ) - H({{\hat x}_k})\\
		% 		= \frac{1}{2}\log \frac{{\det P_k^ - }}{{\det {P_k}}}\\
		% 		= \frac{1}{2}\log \frac{1}{{\det (I - {K_k}{H_k})}}\\
		% 		= \frac{1}{2}\log \frac{1}{{\det [I - P_k^ - H_k^\top{{({H_k}P_k^ - H_k^\top + {H_k}\widehat X_k^ - )}^{ - 1}}{H_k}]}}
		% 	\end{array}
	% \end{equation}
\begin{equation}
	\begin{aligned}
		{\rm{I}}\left( {{z_k},{\bf{x}}\left| {{{\bf{Z}}_{k - 1}}} \right.} \right) =&~ {\rm H}({\bf x} \mid {\bf Z}_{k-1}) - {\rm H}({\bf x} \mid {\bf Z}_{k})\\
		=&~ \left[\frac{N}{2}(1+\log 2 \pi)+\frac{1}{2} \log \left|\widehat{\bf{P}}_{k-1}\right| \right] - \left[   \frac{N}{2}(1+\log 2 \pi)+\frac{1}{2} \log \left|\widehat{\bf{P}}_k\right| \right]\\
		=&~ \frac{1}{2} \log \frac{\left|\widehat{\bf{P}}_{k-1}\right|}{\left|\widehat{\bf{P}}_k\right|}=\frac{1}{2} \log \frac{\left|\widehat{\bf{P}}_{k-1}\right|}{\left|\left(\boldsymbol{I}-\beta {\bf{K}}_k {\bf{h}}_k^{\top}\right) \widehat{\bf{P}}_{k-1}\right|}\\
		=&~ \frac{1}{2}\log \frac{1}{{\left| {{\bf{I}} - \beta {{\bf{K}}_k}{{\bf{h}}_k^\top}} \right|}}.
	\end{aligned}
\end{equation}
where~${\rm H}(\cdot)$ is the information entropy, which can be expressed as $\frac{1}{2}\log {(2\pi e)^N}\det \sum $ for a $N$-variable Gaussian distribution, $\sum$ is the covariance matrix of the corresponding variables, and~$|{\bf P}|$ denotes the determinant of matrix~$\bf P$. By using the property of matrix determinant~$\left|{\bf I} + {\bf u}{\bf v}^\top\right| = 1 + {\bf u}^\top{\bf v}$, the above expression can be further simplified
\begin{equation}
	\begin{aligned}
		{\rm{I}}\left( {{z_k},{\bf{x}}\mid {{\bf{Z}}_{k - 1}}} \right)\\
		=&~ \frac{1}{2}\log \frac{1}{{1 - \beta {\bf{K}}_k^ \top {{\bf{h}}_k}}}\\
		=&~ \frac{1}{2}\log \frac{1}{{1 - {{\left( {{\beta ^2}{{\bf{h}}_k}{{\widehat {\bf{P}}}_{k - 1}}{\bf{h}}_k^ \top  + {R_k}} \right)}^{ - 1}}{\beta ^2}{{\bf{h}}_k}{{\widehat {\bf{P}}}_{k - 1}}{{\bf{h}}_k}}}\\
		=&~ \frac{1}{2}\log \frac{1}{{1 - \frac{{{\beta ^2}{{\bf{h}}_k}{{\widehat {\bf{P}}}_{k - 1}}{\bf{h}}_k^ \top }}{{{\beta ^2}{{\bf{h}}_k}{{\widehat {\bf{P}}}_{k - 1}}{\bf{h}}_k^ \top  + {R_k}}}}}\\
		=&~ \frac{1}{2}\log \frac{{{\beta ^2}{{\bf{h}}_k}{{\widehat {\bf{P}}}_{k - 1}}{\bf{h}}_k^ \top  + {R_k}}}{{{R_k}}}
	\end{aligned}
\end{equation}

The mean CRB is calculated as follows
\begin{equation}
	\begin{aligned}
		{\rm mCRB} =&~ {\rm Tr} ({\widehat{\bf P}}_k) \\
		=&~ \operatorname{Tr}\left(\left(\bf{I}-\beta {\bf{K}}_k {\bf{h}}_k^{\top}\right) \widehat{\bf{P}}_{k-1}\right)\\
		=&~ \operatorname{Tr}\left(\widehat{\bf{P}}_{k-1}\right)-\operatorname{Tr}\left(\beta {\bf{K}}_k {\bf{h}}_k^{\top} \widehat{\bf{P}}_{k-1}\right) \\
		=&~ \operatorname{Tr}\left(\widehat{\bf{P}}_{k-1}\right)-\operatorname{Tr}\left(\beta^2 \widehat{\bf{P}}_{k-1} {\bf{h}}_{k}\left(\beta^2 {\bf{h}}_k^{\top} \widehat{\bf{P}}_{k-1} {\bf{h}}_{k}+R_k\right)^{-1} {\bf{h}}_k^{\top} \widehat{\bf{P}}_{k-1}\right) \\
	\end{aligned}
\end{equation}
where~${\rm Tr}({\bf P})$ denotes the trace of matrix~$\bf P$.
By using the property of matrix trace, the above expression can be rewritten as
\begin{equation} 
	\begin{aligned}
		{\rm mCRB} =&~ \operatorname{Tr}\left(\widehat{\bf{P}}_{k-1}\right)-\beta^2 \operatorname{Tr}\left({\bf{h}}_k^{\top} \widehat{\bf{P}}_{k-1}^2 {\bf{h}}_{k}\left(\beta^2 {\bf{h}}_k^{\top} \widehat{\bf{P}}_{k-1} {\bf{h}}_{k}+R_k\right)^{-1}\right) \\
		=&~ \operatorname{Tr}\left(\widehat{\bf{P}}_{k-1}\right)-\operatorname{Tr}\left(\frac{{\bf{h}}_k^{\top} \widehat{\bf{P}}_{k-1}^2 {\bf{h}}_{k}}{{\bf{h}}_k^{\top} \widehat{\bf{P}}_{k-1} {\bf{h}}_{k}+\frac{R_k}{\beta^2}}\right) \\
		=&~ \operatorname{Tr}\left(\widehat{\bf{P}}_{k-1}\right)-\frac{{\bf{h}}_k^{\top} \widehat{\bf{P}}_{k-1}^2 {\bf{h}}_{k}}{{\bf{h}}_k^{\top} \widehat{\bf{P}}_{k-1} {\bf{h}}_{k}+\frac{R_k}{\beta^2}}
	\end{aligned}
\end{equation}

\section{Proof of Theorem 1}
Here we give the proof of Theorem 1 in the article, which indicates that the optimal encoding pattern of AIME-MI under the total-energy constraint correspond to a scanning point.
\begin{theorem}[Theoretical solution of AIME-MI under the total-energy constraint]\label{thm1}
	Consider the following optimization problem:
	\begin{equation}
		\begin{aligned}
			\max _{\mathbf{h}_k} & \frac{\mathbf{h}_k^{\top} \mathbf{P}_{k-1} \mathbf{h}_k}{\mathbf{h}_k^{\top} \mathbf{x}}, \\
			\rm { s.t. } & \left\|\mathbf{h}_k\right\|_1=C, \mathbf{h}_k \succeq \mathbf{0},
		\end{aligned}
		\label{eq:AIME}
	\end{equation}
	where~$\mathbf{P}_{k-1}$ is a positive definite covariance matrix, $\mathbf{x}$ is a non-negative vector, and~$C>0$ is a constant.
	Then the above problem admits an analytical optimal solution. Specifically, the optimal encoding vector is given by
	\begin{equation}
		\hat{\mathbf{h}}_k=\mathbf{e}_{\hat{i}}, \quad \hat{i}=\arg \max _i \frac{\left(\mathbf{P}_{k-1}\right)_{i i}}{x_i},
	\end{equation}
	where~$\mathbf{e}_{\hat{i}}$ denotes the standard basis vector whose $\hat{i}$-th element is one and all other elements are zero.
\end{theorem}

\begin{proof}[Proof of Theorem~{\upshape\ref{thm1}}]
	We first note that the objective function is a ratio of two positive homogeneous functions of~${\bf h}_k$. Following standard results in fractional programming~\cite{schaible1976fractional}, maximizing the ratio
	\begin{equation}
		\frac{\mathbf{h}_k^{\top} \mathbf{P}_{k-1} \mathbf{h}_k}{\mathbf{h}_k^{\top} \mathbf{x}},
	\end{equation}
	is equivalent to finding the largest scalar~$\alpha$ such that
	\begin{equation}
		\mathbf{h}_k^{\top} \mathbf{P}_{k-1} \mathbf{h}_k - \alpha \mathbf{h}_k^{\top} \mathbf{x} \ge 0
	\end{equation}
	admits a feasible solution of~${\bf h}_k$ under the same constraints.
	This lead to an equivalent parametric optimization problem as
	\begin{equation}
		\max _{\mathbf{h}_k}\left(\mathbf{h}_k^{\top} \mathbf{P}_{k-1} \mathbf{h}_k-\alpha \mathbf{h}_k^{\top} \mathbf{x}\right), \quad \text { s.t. }\left\|\mathbf{h}_k\right\|_1=C,~ \mathbf{h}_k \succeq \mathbf{0},
	\end{equation}
	where 
	\begin{equation}
		\alpha=\max _{\mathbf{h}_k} \frac{\mathbf{h}_k^{\top} \mathbf{P}_{k-1} \mathbf{h}_k}{\mathbf{h}_k^{\top} \mathbf{x}}.
	\end{equation}	 
	This can be further rewritten to a standard quadratic form
	\begin{equation}
		\mathbf{h}_k^{\top} \mathbf{P}_{k-1} \mathbf{h}_k-\alpha \mathbf{h}_k^{\top} \mathbf{x}=\left(\mathbf{h}_k-\frac{\alpha}{2} \mathbf{P}_{k-1}^{-1} \mathbf{x}\right)^{\top} \mathbf{P}_{k-1}\left(\mathbf{h}_k-\frac{\alpha}{2} \mathbf{P}_{k-1}^{-1} \mathbf{x}\right)-\frac{\alpha^2}{4} \mathbf{x}^{\top} \mathbf{P}_{k-1}^{-1} \mathbf{x}.
		\label{eq:quadratic}
	\end{equation}
	Since~$\mathbf{P}_{k-1}$ is positive definite, the objective is convex with respect to~${\bf h}_k$. Under the constraints~$\left\|\mathbf{h}_k\right\|_1=C,$ and~$\mathbf{h}_k \succeq \mathbf{0}$, the feasible region is a simplex, making the extreme points correspond to scaled standard basis vectors~$\left\{C {\bf e}_i\right\}$. Therefore, the maximum of the objective function shown as Eq.~(\ref{eq:quadratic}) is attained at one of the extreme points, implying that the optimal solution~${\hat {\bf h}_k}$ correspond to a standard basis vector.
	
	Substituting~${\bf h}_k = {\bf e}_i$ into the original objective yields
	\begin{equation}
		\frac{\mathbf{h}_k^{\top} \mathbf{P}_{k-1} \mathbf{h}_k}{\mathbf{h}_k^{\top} \mathbf{x}}=\frac{\mathbf{e}_i^{\top} \mathbf{P}_{k-1} \mathbf{e}_i}{\mathbf{e}_i^{\top} \mathbf{x}}=\frac{\left(\mathbf{P}_{k-1}\right)_{i i}}{x_i}.
	\end{equation}
	From this, the optimal index~$\hat i$ can be obtained by
	\begin{equation}
		\hat{i}=\arg \max _i \frac{\left(\mathbf{P}_{k-1}\right)_{i i}}{x_i},
	\end{equation}
	which leads to the optimal~$\hat{\bf h}_k = {\bf e}_{\hat i}$ and complete the proof.
\end{proof}

\section{Definition of the detection signal-to-noise ratio}
Here we would like to discuss about the definition of the detection SNR in GI systems. In the field of GI, the bucket detection SNR~(BSNR)~\cite{du2025information} is calculated as
\begin{equation}
	{\rm BSNR} = 10 \log_{10} \left(\frac{\overline{z^2} - \overline{z}^2}{\sigma^2}\right),
\end{equation}
to measure the level of signal fluctuations. However, this definition only apply to GI systems with randomly fluctuating light fields, and it is not directly related to the detection process.

In the field of signal processing, the SNR is defined as
\begin{equation}
	{\rm SNR} = 10 \log_{10} \left(\frac{P_s}{P_n}\right) = 10 \log_{10} \left(\frac{|z|^2}{|n|^2}\right),
\end{equation}
where~$P_s$ denotes the signal power and~$P_n$ is the noise power.

Different from traditional signal processing which deal with electronic signals, in GI the signal is in fact an optical one. Since the detector directly measures the intensity of light fields, the recorded intensity signal can already be considered as the power. Hence, 
in this study we use the definition as~\cite{bian2024broadband}
\begin{equation} \label{eq: SNR def}
	{\rm SNR} = 10 \log_{10} \left(\frac{\overline{z}}{\sigma}\right),
\end{equation}
where~$\overline{z}$ represents the average value of the detection signal and~$\sigma$ denotes the average noise level.

\section{Detailed comparison results between GI with AIME-$\ell_1$ and fixed point-to-point conventional imaging under different detection SNRs}
Figures~\ref{fig:Exper1-1} to~\ref{fig:Exper1-4} show experimental results of four representative imaging scenes.
\begin{figure}[ht]
	\centering
	\includegraphics[width=0.7\textwidth]{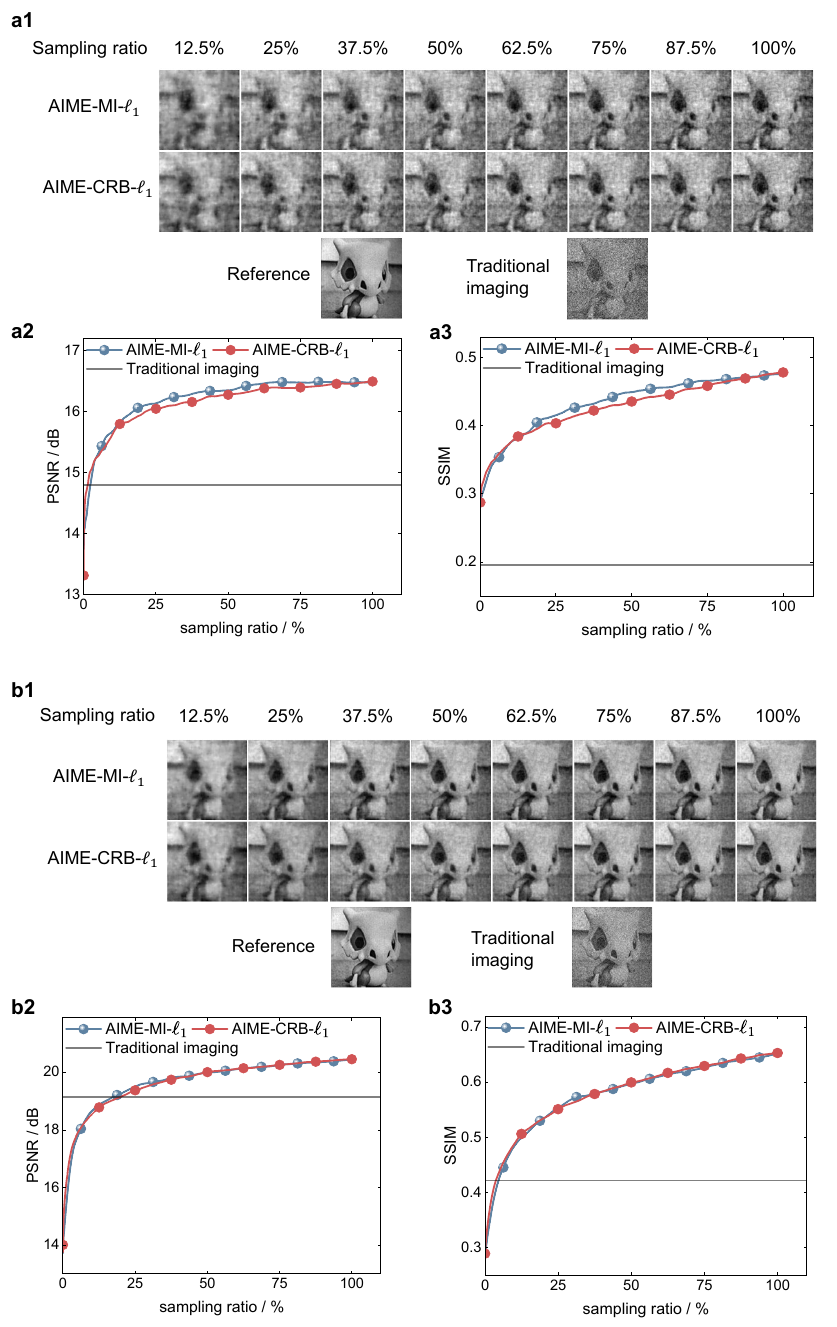}
	\caption{Comparison between GI with AIME-$\ell_1$ and conventional imaging on scene 1 under varying SRs and SNRs, including imaging results and imaging quality indicators. \textbf{a1, a2, a3} are results under~$4.85$ dB SNR, \textbf{b1, b2, b3} are results under~$7.03$ dB SNR.}
	\label{fig:Exper1-1}
\end{figure}

\begin{figure}[ht]
	\centering
	\includegraphics[width=0.78\textwidth]{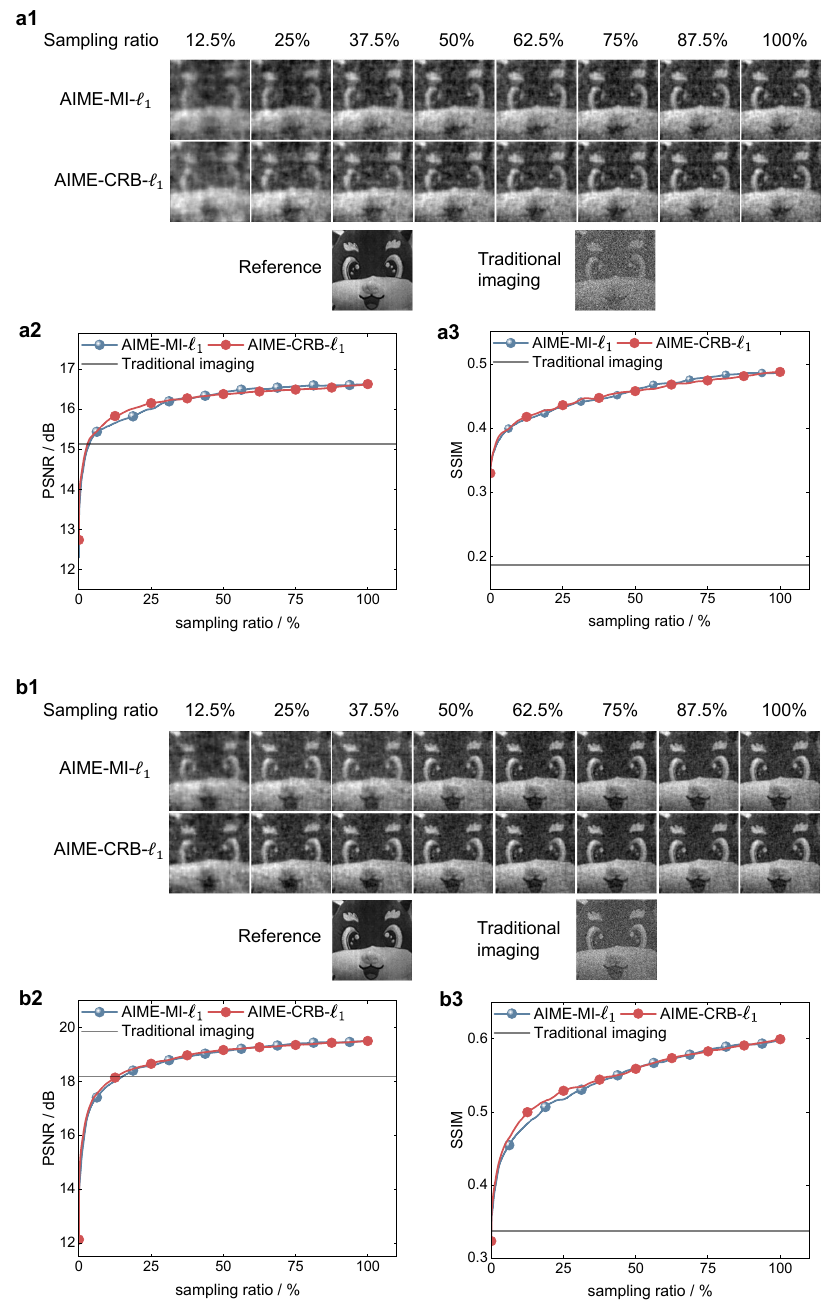}
	\caption{Comparison between GI with AIME-$\ell_1$ and conventional imaging on scene 1 under varying SRs and SNRs, including imaging results and imaging quality indicators. \textbf{a1, a2, a3} are results under~$3.28$ dB SNR, \textbf{b1, b2, b3} are results under~$5.79$ dB SNR.}
	\label{fig:Exper1-2}
\end{figure}

\begin{figure}[ht]
	\centering
	\includegraphics[width=0.78\textwidth]{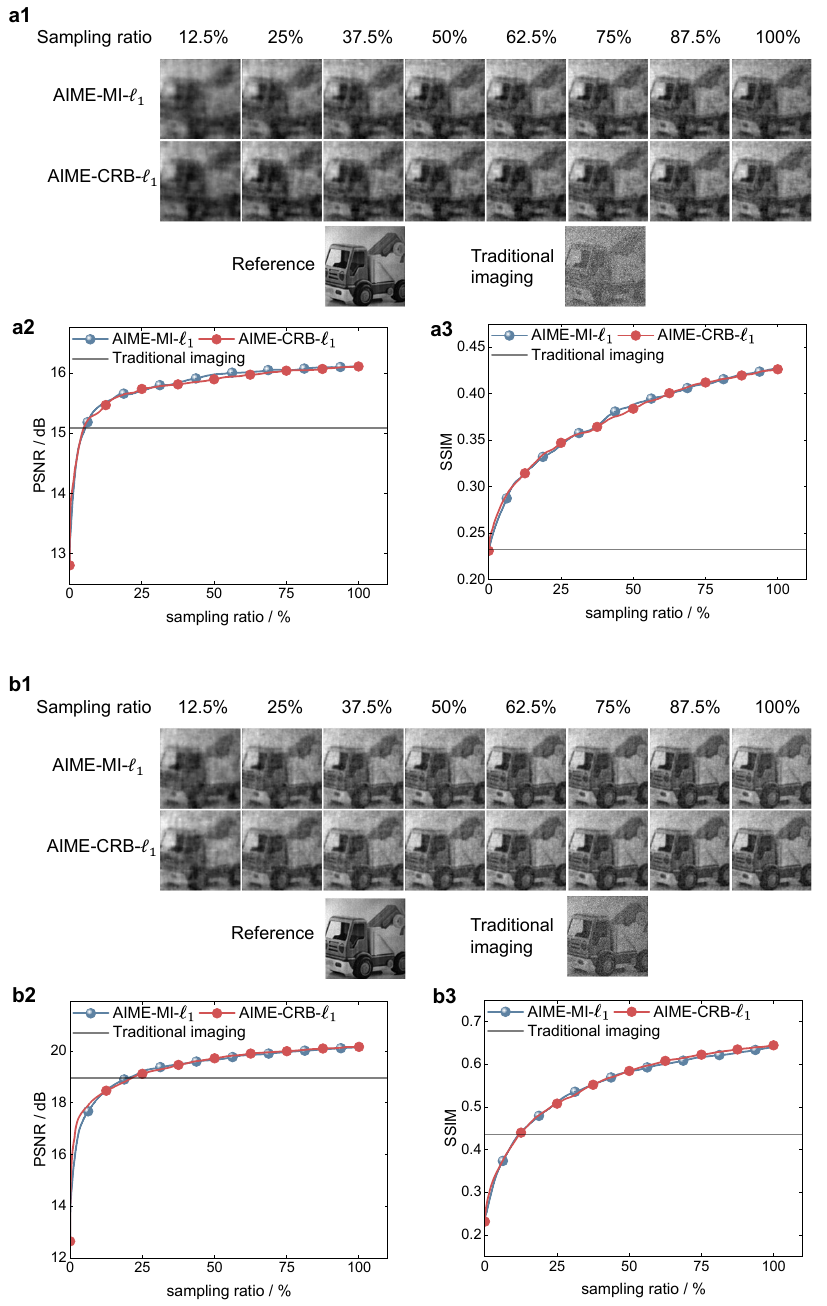}
	\caption{Comparison between GI with AIME-$\ell_1$ and conventional imaging on scene 1 under varying SRs and SNRs, including imaging results and imaging quality indicators. \textbf{a1, a2, a3} are results under~$3.89$ dB SNR, \textbf{b1, b2, b3} are results under~$5.79$ dB SNR.}
	\label{fig:Exper1-3}
\end{figure}

\begin{figure}[ht]
	\centering
	\includegraphics[width=0.78\textwidth]{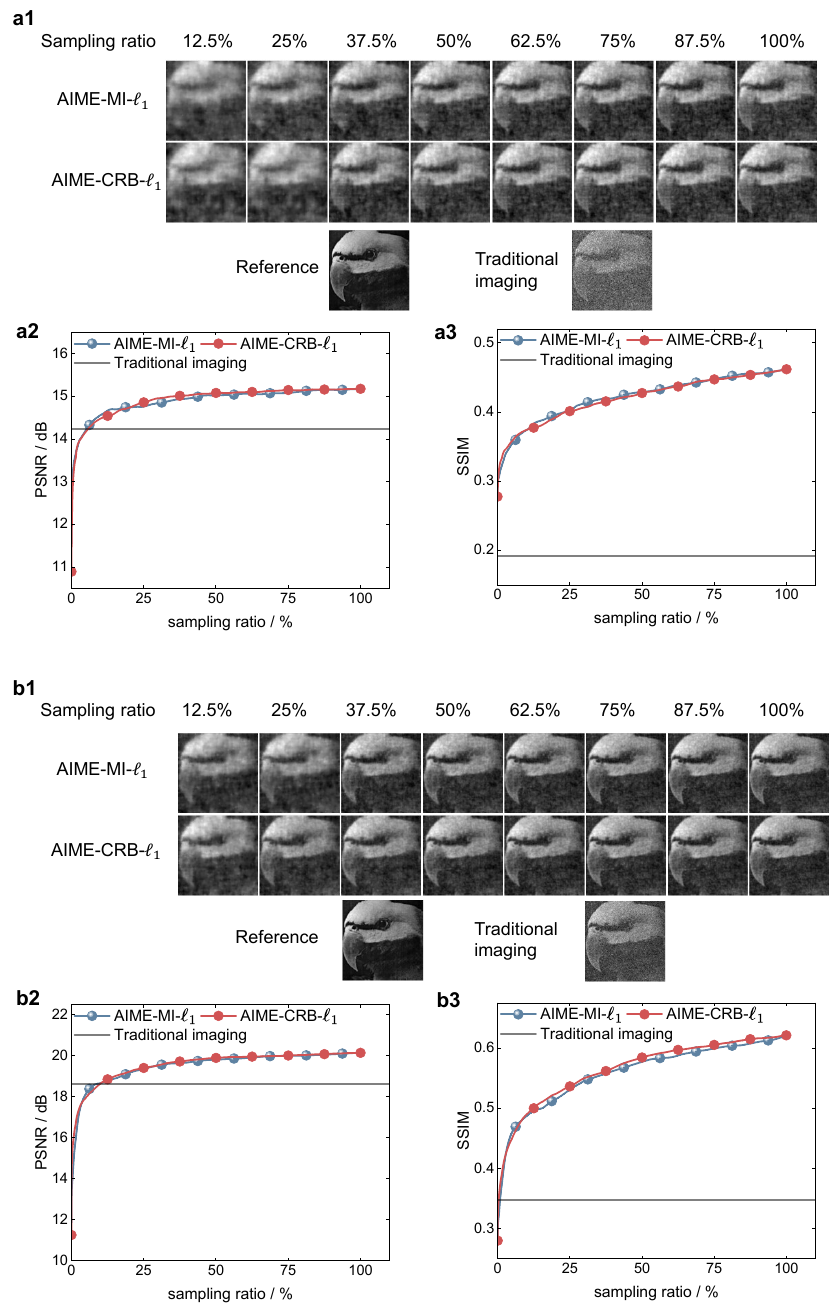}
	\caption{Comparison between GI with AIME-$\ell_1$ and conventional imaging on scene 1 under varying SRs and SNRs, including imaging results and imaging quality indicators. \textbf{a1, a2, a3} are results under~$2.09$ dB SNR, \textbf{b1, b2, b3} are results under~$4.30$ dB SNR.}
	\label{fig:Exper1-4}
\end{figure}

\FloatBarrier
\section{Detailed comparison between GI with AIME and existing encoding patterns}
Figures~\ref{fig:Exper2-1} to~\ref{fig:Exper2-4} show experimental results of four representative imaging scenes for comparing GI with AIME and existing encoding patterns.
\begin{figure}[ht]
	\centering
	\includegraphics[width=0.78\textwidth]{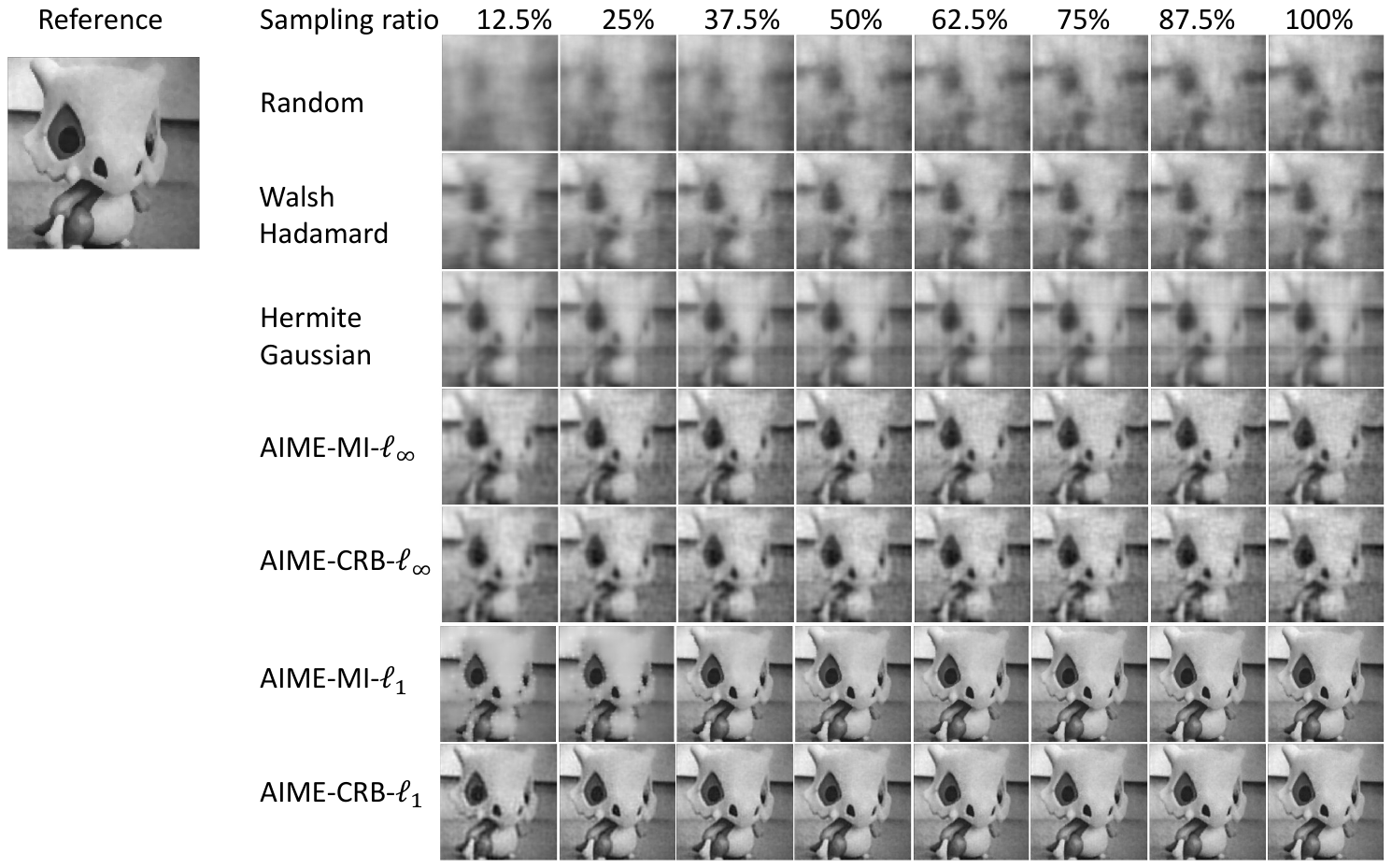}
	\caption{Experimental imaging results of scene 1 under different sampling ratios for comparing GI with AIME and existing encoding patterns.}
	\label{fig:Exper2-1}
\end{figure}

\begin{figure}[ht]
	\centering
	\includegraphics[width=0.78\textwidth]{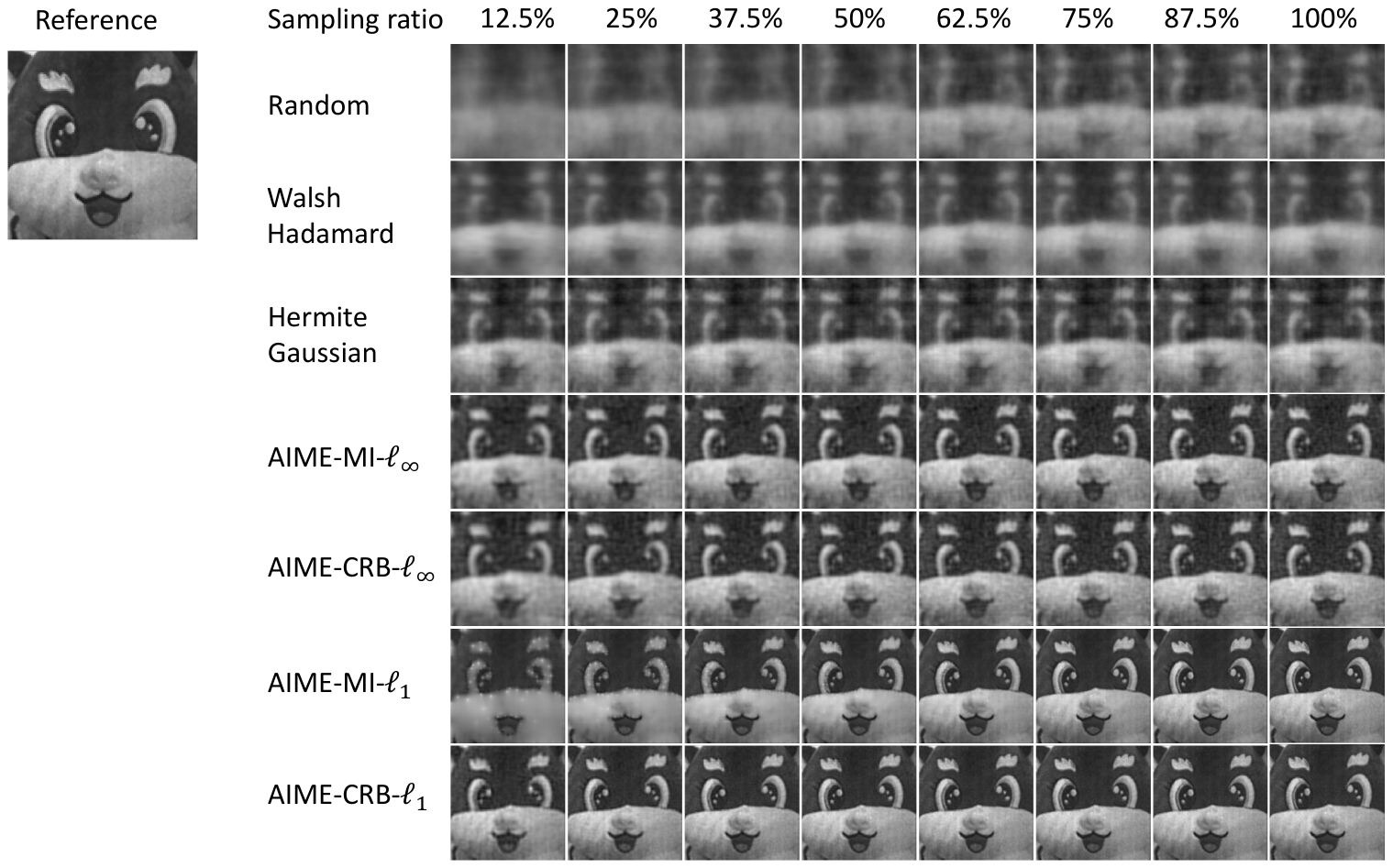}
	\caption{Experimental imaging results of scene 2 under different sampling ratios for comparing GI with AIME and existing encoding patterns.}
	\label{fig:Exper2-2}
\end{figure}

\begin{figure}[ht]
	\centering
	\includegraphics[width=0.85\textwidth]{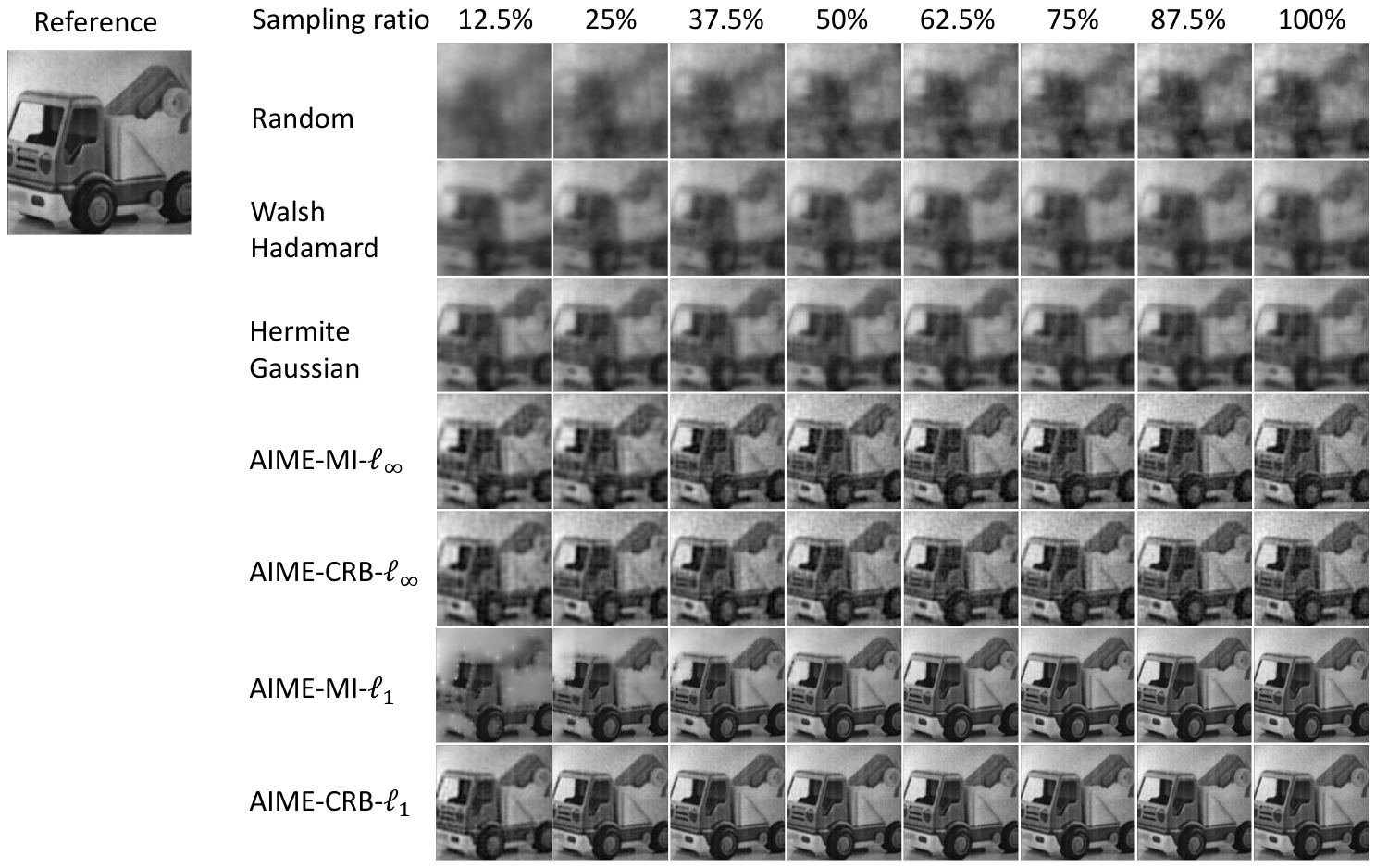}
	\caption{Experimental imaging results of scene 3 under different sampling ratios for comparing GI with AIME and existing encoding patterns.}
	\label{fig:Exper2-3}
\end{figure}

\begin{figure}[ht]
	\centering
	\includegraphics[width=0.85\textwidth]{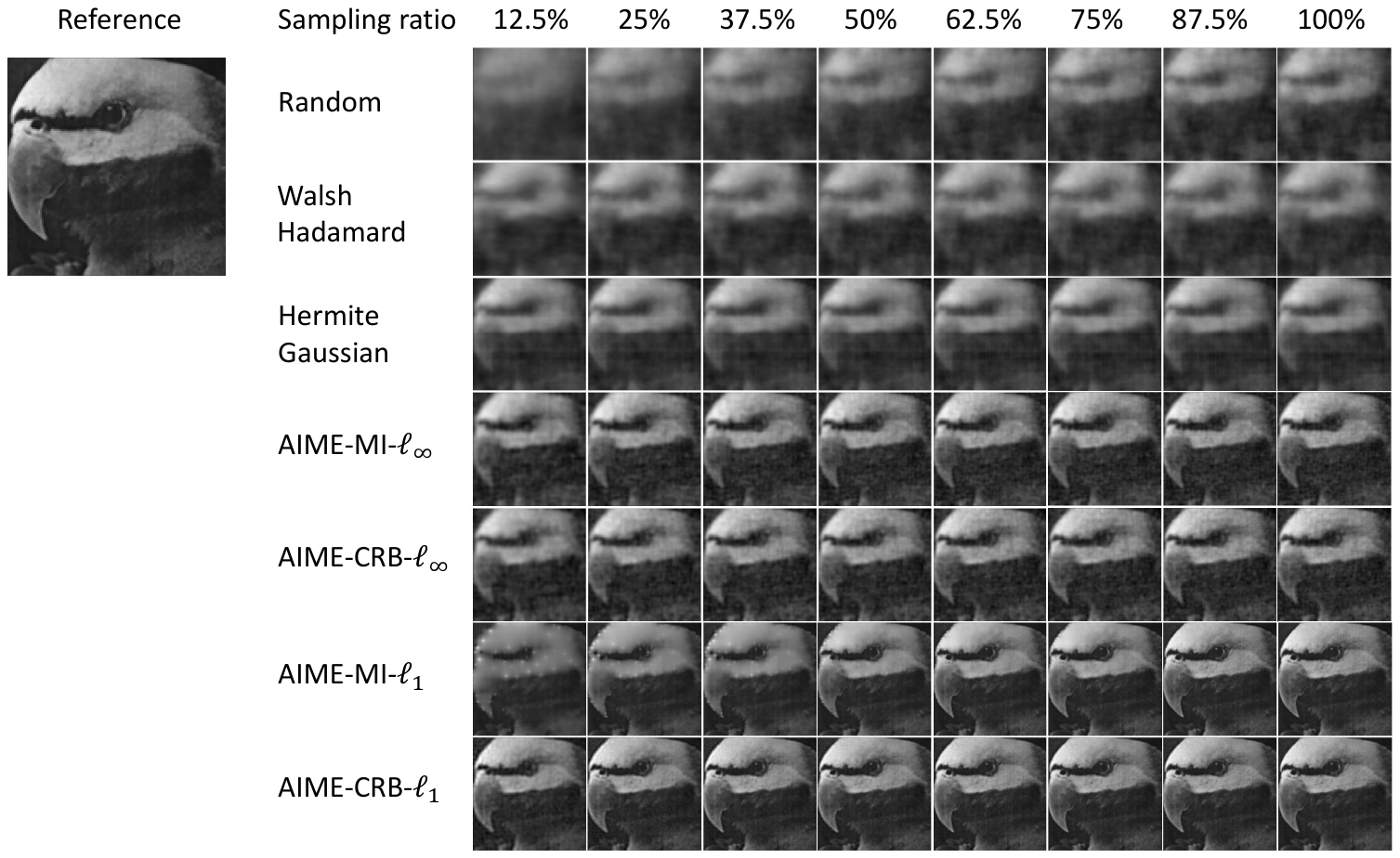}
	\caption{Experimental imaging results of scene 4 under different sampling ratios for comparing GI with AIME and existing encoding patterns.}
	\label{fig:Exper2-4}
\end{figure}

\FloatBarrier
\section{Information-optimal encoding under different SNRs}
Figures~\ref{fig:PattSNR-l1} and~\ref{fig:PattSNR} show AIME results corresponding to the same imaging scene \textit{pepper}  under three different SNRs, which indicates that AIME can adaptively deal with different noise levels and constraints, resulting in different encoding patterns.
\begin{figure}[ht]
	\centering
	\includegraphics[width=0.9\textwidth]{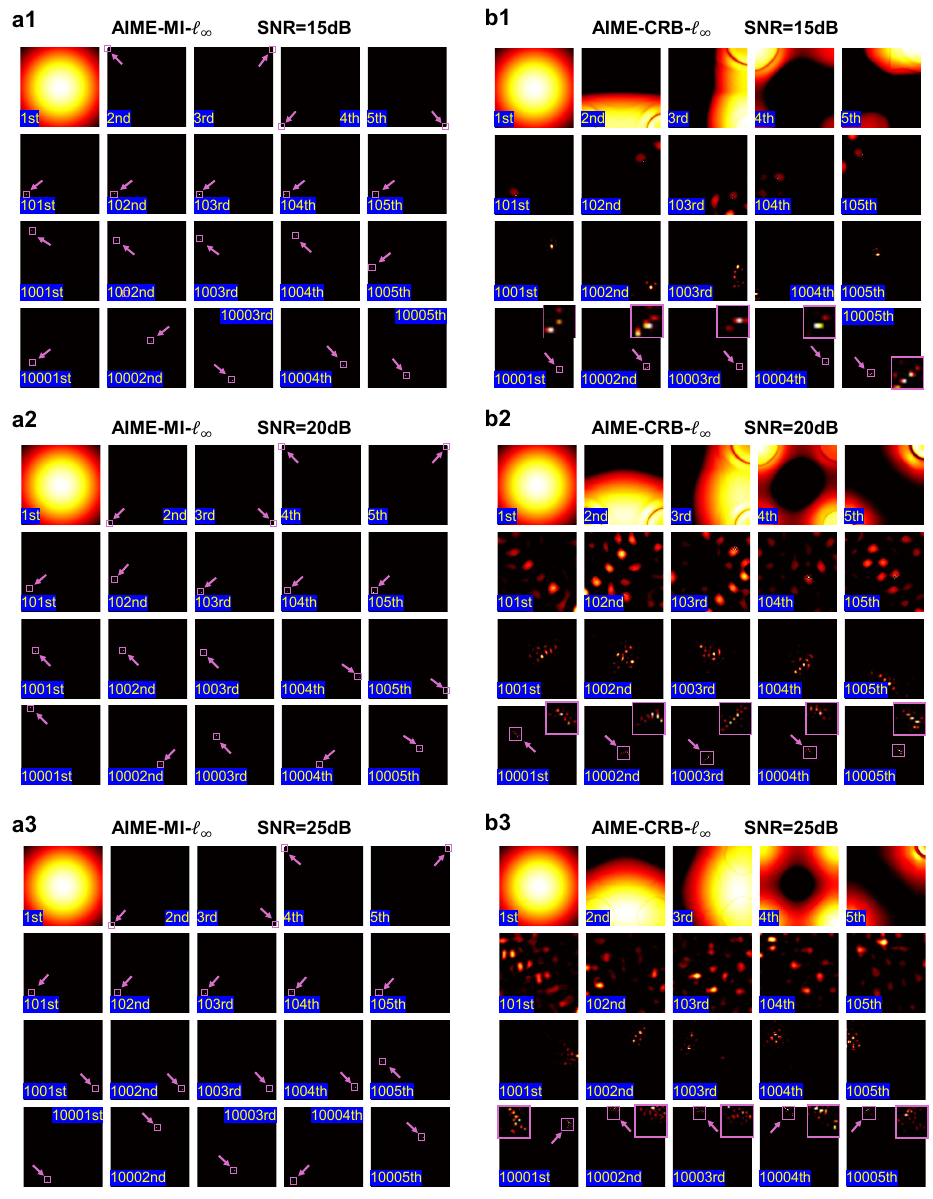}
	\caption{Information-optimal encoding from AIME-$\ell_1$ under three different SNRs.}
	\label{fig:PattSNR-l1}
\end{figure}

\begin{figure}[ht]
	\centering
	\includegraphics[width=0.9\textwidth]{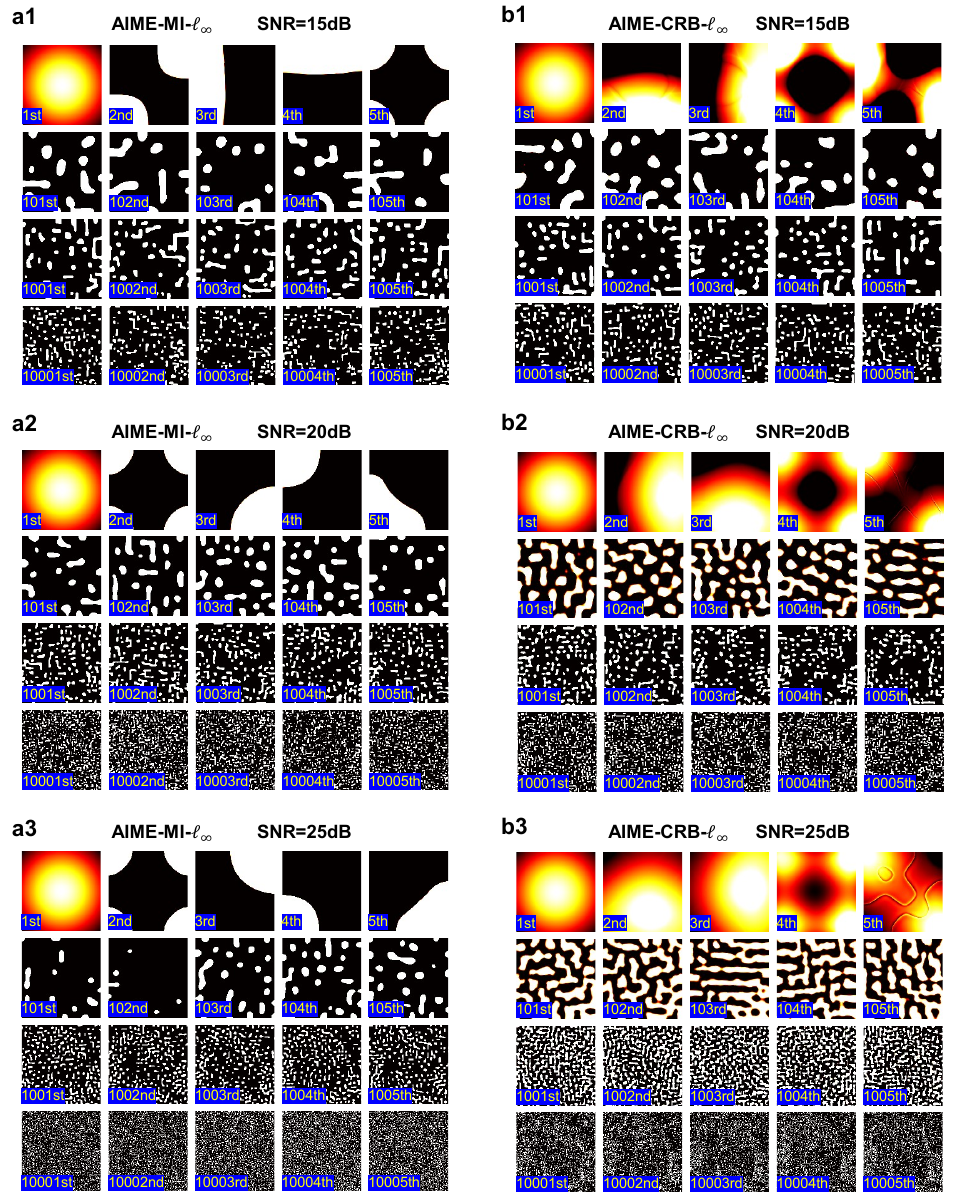}
	\caption{Information-optimal encoding from AIME-$\ell_{\infty}$ under three different SNRs.}
	\label{fig:PattSNR}
\end{figure}
\FloatBarrier
\end{appendices}
\bibliography{article-supp-bibliography}

@book{barrett2013foundations,
  title={Foundations of image science},
  author={Barrett, Harrison H and Myers, Kyle J},
  year={2013},
  publisher={John Wiley \& Sons},
  address={Hoboken, NJ, USA}
}

@article{wei2016mutual,
  title={Mutual information, Fisher information, and efficient coding},
  author={Wei, Xue-Xin and Stocker, Alan A},
  journal={Neural Computation},
  volume={28},
  number={2},
  pages={305--326},
  year={2016},
  publisher={MIT Press One Rogers Street, Cambridge, MA 02142-1209, USA journals-info~…}
}

@article{hu2022ghost,
  title={On ghost imaging studies for information optical imaging},
  author={Hu, Chenyu and Han, Shensheng},
  journal={Applied Sciences},
  volume={12},
  number={21},
  pages={10981},
  year={2022},
  publisher={MDPI}
}

@article{li2019single,
	title={Single-frame wide-field nanoscopy based on ghost imaging via sparsity constraints},
	author={Li, Wenwen and Tong, Zhishen and Xiao, Kang and Liu, Zhentao and Gao, Qi and Sun, Jing and Liu, Shupeng and Han, Shensheng and Wang, Zhongyang},
	journal={Optica},
	volume={6},
	number={12},
	pages={1515--1523},
	year={2019},
	publisher={Optica Publishing Group}
}

@article{tong2020breaking,
  title={Breaking Rayleigh's Criterion via Discernibility in High-Dimensional Light-Field Space with Snapshot Ghost Imaging},
  author={Tong, Zhishen and Liu, Zhentao and Wang, Jian and Shen, Xia and Han, Shensheng},
  journal={arXiv preprint arXiv:2004.00135},
  year={2020}
}

@article{liu2020spectral,
  title={Spectral ghost imaging camera with super-Rayleigh modulator},
  author={Liu, Shengying and Liu, Zhentao and Hu, Chenyu and Li, Enrong and Shen, Xia and Han, Shensheng},
  journal={Optics Communications},
  volume={472},
  pages={126017},
  year={2020},
  publisher={Elsevier}
}

@article{chu2021spectral,
  title={Spectral polarization camera based on ghost imaging via sparsity constraints},
  author={Chu, Chunyan and Liu, Shengying and Liu, Zhentao and Hu, Chenyu and Zhao, Yuejin and Han, Shensheng},
  journal={Applied Optics},
  volume={60},
  number={16},
  pages={4632--4638},
  year={2021},
  publisher={Optical Society of America}
}

@article{deng2017performance,
  title={Performance analysis of ghost imaging lidar in background light environment},
  author={Deng, Chenjin and Pan, Long and Wang, Chenglong and Gao, Xin and Gong, Wenlin and Han, Shensheng},
  journal={Photonics Research},
  volume={5},
  number={5},
  pages={431--435},
  year={2017},
  publisher={Chinese Laser Press and Optical Society of America}
}

@article{fisher1922mathematical,
  title={On the mathematical foundations of theoretical statistics},
  author={Fisher, Ronald A},
  journal={Philosophical transactions of the Royal Society of London. Series A, containing papers of a mathematical or physical character},
  volume={222},
  number={594-604},
  pages={309--368},
  year={1922},
  publisher={The Royal Society London}
}

@article{enrong2013mutual,
  title={Mutual information of ghost imaging systems},
  author={Li, Enrong and Chen, Mingliang and Gong, Wenlin and Yu, Hong and Han, Shensheng},
  journal={Acta Optica Sinica},
  volume={33},
  number={12},
  pages={1211003},
  year={2013}
}

@article{li2017negative,
  title={Negative exponential behavior of image mutual information for pseudo-thermal light ghost imaging: observation, modeling, and verification},
  author={Li, Junhui and Luo, Bin and Yang, Dongyue and Yin, Longfei and Wu, Guohua and Guo, Hong},
  journal={Science Bulletin},
  volume={62},
  number={10},
  pages={717--723},
  year={2017},
  publisher={Elsevier}
}

@article{hu2021correspondence,
  title={Correspondence Fourier-transform ghost imaging},
  author={Hu, Chenyu and Zhu, Ruiguo and Yu, Hong and Han, Shensheng},
  journal={Physical Review A},
  volume={103},
  number={4},
  pages={043717},
  year={2021},
  publisher={APS}
}

@book{eldar2012compressed,
  title={Compressed Sensing: Theory and Applications},
  author={Eldar, Yonina C and Kutyniok, Gitta},
  year={2012},
  publisher={Cambridge University Press},
  address   = {Cambridge, UK}
}

@article{tropp2004greed,
  title={Greed is good: Algorithmic results for sparse approximation},
  author={Tropp, Joel A},
  journal={IEEE Transactions on Information theory},
  volume={50},
  number={10},
  pages={2231--2242},
  year={2004},
  publisher={IEEE}
}

@article{xu2015optimization,
  title={Optimization of speckle patterns in ghost imaging via sparse constraints by mutual coherence minimization},
  author={Xu, Xuyang and Li, Enrong and Shen, Xia and Han, Shensheng},
  journal={Chinese Optics Letters},
  volume={13},
  number={7},
  pages={071101},
  year={2015},
  publisher={Chinese Optical Society}
}

@article{czajkowski2018single,
  title={Single-pixel imaging with Morlet wavelet correlated random patterns},
  author={Czajkowski, Krzysztof M and Pastuszczak, Anna and Koty{\'n}ski, Rafa{\l}},
  journal={Scientific Reports},
  volume={8},
  number={1},
  pages={466},
  year={2018},
  publisher={Nature Publishing Group UK London}
}

@article{hu2019optimization,
  title={Optimization of light fields in ghost imaging using dictionary learning},
  author={Hu, Chenyu and Tong, Zhishen and Liu, Zhentao and Huang, Zengfeng and Wang, Jian and Han, Shensheng},
  journal={Optics Express},
  volume={27},
  number={20},
  pages={28734--28749},
  year={2019},
  publisher={Optica Publishing Group}
}

@article{arias2012fundamental,
  title={On the fundamental limits of adaptive sensing},
  author={Arias-Castro, Ery and Candes, Emmanuel J and Davenport, Mark A},
  journal={IEEE Transactions on Information Theory},
  volume={59},
  number={1},
  pages={472--481},
  year={2012},
  publisher={IEEE}
}

@article{shensheng2022review,
  title={Review, current status and prospect of researches on information optical imaging},
  author={Shensheng, Han and Chenyu, Hu},
  journal={Infrared and Laser Engineering},
  volume={51},
  number={1},
  pages={20220017--1},
  year={2022},
  publisher={Editorial Office of Journal of Infrared and Laser Engineering}
}

@article{du2025information,
  title={Information-quantitative evaluation of linear computational imaging and application in ghost imaging},
  author={Du, Long-Kun and Hu, Chenyu and Nie, Zhen-Wu and Chang, Chen and Sun, Shuai and Liu, Shuang and Deng, Chenjin and Bo, Zunwang and Liu, Wei-Tao and Han, Shensheng},
  journal={Optics \& Laser Technology},
  volume={192},
  pages={113893},
  year={2025},
  publisher={Elsevier}
}

@article{tong2025single,
  title={Single-shot super-resolution imaging via discernibility in the high-dimensional light-field space based on ghost imaging},
  author={Tong, Zhishen and Hu, Chenyu and Wang, Jian and Zhu, Youheng and Shen, Xia and Liu, Zhentao and Han, Shensheng},
  journal={Photonics Research},
  volume={13},
  number={6},
  pages={1709--1725},
  year={2025}
}

@article{zhang2018tabletop,
  title={Tabletop x-ray ghost imaging with ultra-low radiation},
  author={Zhang, Ai-Xin and He, Yu-Hang and Wu, Ling-An and Chen, Li-Ming and Wang, Bing-Bing},
  journal={Optica},
  volume={5},
  number={4},
  pages={374--377},
  year={2018},
  publisher={OSA}
}

@article{pan2021micro,
  title={Micro-Doppler effect based vibrating object imaging of coherent detection GISC lidar},
  author={Pan, Long and Wang, Yiqun and Deng, Chenjin and Gong, Wenlin and Bo, Zunwang and Han, Shensheng},
  journal={Optics Express},
  volume={29},
  number={26},
  pages={43022--43031},
  year={2021},
  publisher={Optical Society of America}
}

@article{chen2025multicolor,
  title={Multicolor Super-Resolution Structured Illumination Microscopy Based on Snapshot Spectral Ghost Imaging via Sparsity Constraints},
  author={Chen, Li and Wang, Pengwei and Liu, Zhentao and Wu, Jianrong and Han, Shensheng},
  journal={ACS Photonics},
  volume={12},
  number={7},
  pages={3565-3573},
  year={2025},
  publisher={ACS Publications}
}

@article{eshun20253d,
  title={3D quantum ghost imaging microscope},
  author={Eshun, Audrey and Davenport, Dominique and Demory, Brandon and Kiannejad, Shervin and Mos, Paul and Lin, Yang and Bond, Tiziana and Rushford, Michael C and Nuccio, Erin E and Samo, Ty J and others},
  journal={Optica},
  volume={12},
  number={7},
  pages={1109--1112},
  year={2025},
  publisher={Optica Publishing Group}
}

@article{huang2024compressed,
  title={Compressed Hermite--Gaussian differential single-pixel imaging},
  author={Huang, Guancheng and Shuai, Yong and Ji, Yu and Zhou, Xuyang and Li, Qi and Liu, Wei and Gao, Bin and Liu, Shutian and Liu, Zhengjun and Li, Yutong},
  journal={Applied Physics Letters},
  volume={124},
  number={11},
  year={2024},
  publisher={AIP Publishing}
}

@book{blahut1987principles,
  title={Principles and practice of information theory},
  author={Blahut, Richard E},
  year={1987},
  publisher={Addison-Wesley Longman Publishing Co., Inc.},
  address={Reading, MA, USA}
}

@article{yuan2021single,
  title={Single-pixel neutron imaging with artificial intelligence: Breaking the barrier in multi-parameter imaging, sensitivity, and spatial resolution},
  author={Yuan, Xin and Han, Shensheng},
  journal={The Innovation},
  volume={2},
  number={2},
  year={2021},
  publisher={Elsevier}
}

@article{he2021single,
  title={Single-pixel imaging with neutrons},
  author={He, Yu-Hang and Huang, Yi-Yi and Zeng, Zhi-Rong and Li, Yi-Fei and Tan, Jun-Hao and Chen, Li-Ming and Wu, Ling-An and Li, Ming-Fei and Quan, Bao-Gang and Wang, Song-Lin and others},
  journal={Science Bulletin},
  volume={66},
  number={2},
  pages={133--138},
  year={2021},
  publisher={Elsevier}
}

@article{wu2025image,
  title={Image-free cross-species pose estimation via an ultra-low sampling rate single-pixel camera},
  author={Wu, Xin and Zhou, Cheng and Li, Binyu and Huang, Jipeng and Meng, Yanli and Song, Lijun and Han, Shensheng},
  journal={Chinese Optics Letters},
  volume={23},
  number={9},
  pages={091101},
  year={2025},
  publisher={Chinese Laser Press}
}

@article{abbas2025target,
  title={Target recognition in ghost imaging from traditional to advance; a brief review},
  author={Abbas, Ayesha and Mu, Jianbang and Mengyue, Zhang and Cao, Jie and Zhang, Xiaonan},
  journal={Optics \& Laser Technology},
  volume={184},
  pages={112450},
  year={2025},
  publisher={Elsevier}
}

@article{moreau2018ghost,
  title={Ghost imaging using optical correlations},
  author={Moreau, Paul-Antoine and Toninelli, Ermes and Gregory, Thomas and Padgett, Miles J},
  journal={Laser \& Photonics Reviews},
  volume={12},
  number={1},
  pages={1700143},
  year={2018},
  publisher={Wiley Online Library}
}

@article{malloy2014near,
  title={Near-optimal adaptive compressed sensing},
  author={Malloy, Matthew L and Nowak, Robert D},
  journal={IEEE Transactions on Information Theory},
  volume={60},
  number={7},
  pages={4001--4012},
  year={2014},
  publisher={IEEE}
}

@article{torralba2003statistics,
  title={Statistics of Natural Image Categories},
  author={Torralba, Antonio and Oliva, Aude},
  journal={Network: Computation in Neural Systems},
  volume={14},
  number={3},
  pages={391},
  year={2003},
  publisher={IOP Publishing}
}

@article{bogdanov2025ghost,
  title={Ghost Imaging with Free Electron-Photon Pairs},
  author={Bogdanov, Sergei and Preimesberger, Alexander and Mishra, Harsh and Hornof, Dominik and Spielauer, Thomas and Thajer, Florian and Maurer, Max and Falb, Pia and St{\"o}ger, Leo and Schachinger, Thomas and others},
  journal={arXiv preprint arXiv:2509.14950},
  year={2025}
}

@inproceedings{elata2024adaptive,
  title={Adaptive compressed sensing with diffusion-based posterior sampling},
  author={Elata, Noam and Michaeli, Tomer and Elad, Michael},
  booktitle={European Conference on Computer Vision},
  pages={290--308},
  year={2024},
  organization={Springer}
}

@article{huang2020ghost,
  title={Ghost imaging for detecting trembling with random temporal changing},
  author={Huang, Xianwei and Nan, Suqin and Tan, Wei and Bai, Yanfeng and Fu, Xiquan},
  journal={Optics Letters},
  volume={45},
  number={6},
  pages={1354--1357},
  year={2020},
  publisher={Optical Society of America}
}

@article{huang2021ghost,
  title={Ghost imaging influenced by a supersonic wind-induced random environment},
  author={Huang, Xianwei and Nan, Suqin and Tan, Wei and Bai, Yanfeng and Fu, Xiquan},
  journal={Optics Letters},
  volume={46},
  number={5},
  pages={1009--1012},
  year={2021},
  publisher={Optical Society of America}
}

@article{liu2025comprehensive,
  title={Comprehensive compensation of real-world degradations for robust single-pixel imaging},
  author={Liu, Zonghao and Yang, Bohan and Zhang, Yifei and Shen, Junfei and Yuan, Xin and Chen, Mu Ku and Liu, Fei and Geng, Zihan},
  journal={Light: Science \& Applications},
  volume={14},
  number={1},
  pages={365},
  year={2025},
  publisher={Nature Publishing Group UK London}
}

@article{arias2011noise,
  title={Noise folding in compressed sensing},
  author={Arias-Castro, Ery and Eldar, Yonina C},
  journal={IEEE Signal Processing Letters},
  volume={18},
  number={8},
  pages={478--481},
  year={2011},
  publisher={IEEE}
}

@article{donoho2006compressed,
  title={Compressed sensing},
  author={Donoho, David L},
  journal={IEEE Transactions on information theory},
  volume={52},
  number={4},
  pages={1289--1306},
  year={2006},
  publisher={IEEE}
}

@book{hyvarinen2009natural,
  title={Natural image statistics: A probabilistic approach to early computational vision.},
  author={Hyv{\"a}rinen, Aapo and Hurri, Jarmo and Hoyer, Patrick O},
  volume={39},
  year={2009},
  publisher={Springer Science and Business Media},
  address = {New York, NY, USA}
}

@article{pinkard2024information,
  title={Information-driven design of imaging systems},
  author={Pinkard, Henry and Kabuli, Leyla and Markley, Eric and Chien, Tiffany and Jiao, Jiantao and Waller, Laura},
  journal={arXiv preprint arXiv:2405.20559},
  year={2024}
}

@article{candes2006robust,
  author={Candes, E.J. and Romberg, J. and Tao, T.},
  journal={IEEE Transactions on Information Theory}, 
  title={Robust uncertainty principles: exact signal reconstruction from highly incomplete frequency information}, 
  year={2006},
  volume={52},
  number={2},
  pages={489-509},
 }

@article{stantchev2020real,
  title={Real-time terahertz imaging with a single-pixel detector},
  author={Stantchev, Rayko Ivanov and Yu, Xiao and Blu, Thierry and Pickwell-MacPherson, Emma},
  journal={Nature communications},
  volume={11},
  number={1},
  pages={2535},
  year={2020},
  publisher={Nature Publishing Group UK London}
}

@article{wu2020online,
  title={Online adaptive computational ghost imaging},
  author={Wu, Heng and Wang, Ruizhou and Huang, Zepeng and Xiao, Huapan and Liang, Jian and Wang, Daodang and Tian, Xiaobo and Wang, Tao and Cheng, Lianglun},
  journal={Optics and Lasers in Engineering},
  volume={128},
  pages={106028},
  year={2020},
  publisher={Elsevier}
}

@article{wang2023dual,
  title={Dual-mode adaptive-SVD ghost imaging},
  author={Wang, Dajing and Liu, Baolei and Song, Jiaqi and Wang, Yao and Shan, Xuchen and Zhong, Xiaolan and Wang, Fan},
  journal={Optics Express},
  volume={31},
  number={9},
  pages={14225--14239},
  year={2023},
  publisher={Optica Publishing Group}
}

@book{land2012animal,
  title={Animal eyes},
  author={Land, Michael F and Nilsson, Dan-Eric},
  year={2012},
  publisher={Oxford University Press},
  address={Oxford, UK}
}

@article{zhou2024theoretical,
  title={Theoretical and experimental analysis on ghost imaging with channel coding theorem},
  author={Zhou, Yu and Liu, Jianbin and Zheng, Huaibin and Chen, Hui and He, Yuchen and Li, Fuli and Xu, Zhuo},
  journal={Optics Express},
  volume={32},
  number={25},
  pages={43911--43928},
  year={2024},
  publisher={Optica Publishing Group}
}

@article{schaible1976fractional,
  title={Fractional programming. II, on Dinkelbach's algorithm},
  author={Schaible, Siegfried},
  journal={Management science},
  volume={22},
  number={8},
  pages={868--873},
  year={1976},
  publisher={INFORMS}
}

@article{bian2024broadband,
  title={A broadband hyperspectral image sensor with high spatio-temporal resolution},
  author={Bian, Liheng and Wang, Zhen and Zhang, Yuzhe and Li, Lianjie and Zhang, Yinuo and Yang, Chen and Fang, Wen and Zhao, Jiajun and Zhu, Chunli and Meng, Qinghao and others},
  journal={Nature},
  volume={635},
  number={8037},
  pages={73--81},
  year={2024},
  publisher={Nature Publishing Group UK London}
}


\end{document}